\documentclass[aps,pra,reprint,twocolumn,floatfix,superscriptaddress,notitlepage,usenames,dvipsnames,svgnames,table,footinbib,longbibliography]{revtex4-1}
\usepackage[utf8]{inputenc}
\usepackage[T1]{fontenc}
\usepackage{graphicx}
\usepackage{grffile}
\usepackage{wrapfig}
\usepackage{rotating}
\usepackage[normalem]{ulem}
\usepackage{amsmath}
\usepackage{textcomp}
\usepackage{amssymb}
\usepackage{capt-of}

\usepackage{parskip}
\usepackage{mathrsfs}
\usepackage[margin= 0.75in]{geometry}
\usepackage[braket, qm]{qcircuit}

\usepackage[english]{babel}
\usepackage{letltxmacro}
\LetLtxMacro{\ORIGselectlanguage}{\selectlanguage}
\makeatletter
\DeclareRobustCommand{\selectlanguage}[1]{%
  \@ifundefined{alias@\string#1}
    {\ORIGselectlanguage{#1}}
    {\begingroup\edef\x{\endgroup
      \noexpand\ORIGselectlanguage{\@nameuse{alias@#1}}}\x}%
}
\newcommand{\definelanguagealias}[2]{%
  \@namedef{alias@#1}{#2}%
}
\makeatother
\definelanguagealias{en}{english}

\usepackage{times}
\usepackage{bbm}
\usepackage{etoolbox}
\usepackage{tcolorbox} 
\usepackage{fleqn} 
\usepackage{tikz}
\usepackage{amsthm}
\usepackage{amssymb}
\usetikzlibrary{arrows.meta}
\usepackage{mathtools}

\usepackage{bm}
\usepackage{dsfont}
\usepackage[font=small,labelfont=bf,justification=justified,format=plain]{caption}
\usepackage{subcaption}

\usepackage{thmtools}
\usepackage{thm-restate}
\newtheorem{thm}{Theorem}
\newtheorem{prop}[thm]{Proposition}

\definecolor{blue-violet}{rgb}{0.54, 0.17, 0.89}
\usepackage{hyperref}
\hypersetup{
 pdfauthor={The quantum coherent people},
 pdftitle={Open OTOC},
 pdflang={English},colorlinks=true,linkcolor=RubineRed,citecolor=blue-violet,urlcolor=Cerulean} 
\usepackage[capitalise]{cleveref}

\begin{document}
\title{Information Scrambling and Chaos in Open Quantum Systems}

\author{Paolo Zanardi}
\email [e-mail: ]{zanardi@usc.edu}

\author{Namit Anand}
\email [e-mail: ]{namitana@usc.edu}

\affiliation{Department of Physics and Astronomy, and Center for Quantum Information Science and Technology, University of Southern California, Los Angeles, California 90089-0484, USA}

\date{\today}

\begin{abstract}
Out-of-time-ordered correlators (OTOCs) have been extensively used over the last few years to study information scrambling and quantum chaos in many-body systems. In this paper, we extend the formalism of the averaged bipartite OTOC of Styliaris \textit{et al} \href{https://journals.aps.org/prl/abstract/10.1103/PhysRevLett.126.030601}{[Phys. Rev. Lett. \textbf{126}, 030601 (2021)]} to the case of open quantum systems. The dynamics is no longer unitary but it is described by more general quantum channels (trace preserving, completely positive maps). This ``open bipartite OTOC'' can be treated in an exact analytical fashion and is shown to amount to a   distance between two quantum channels. Moreover, our analytical form unveils competing entropic contributions from information scrambling and environmental decoherence such that the latter can obfuscate the former. To elucidate this subtle interplay we analytically study special classes of quantum channels, namely, dephasing channels, entanglement-breaking channels, and others.  Finally, as a physical application we numerically study  dissipative many-body spin-chains and show how the competing entropic effects  can be used to differentiate  between integrable and chaotic regimes.

\end{abstract}
\maketitle

\section{Introduction}
\label{sec:introduction}
Many-body quantum chaos has witnessed a renaissance in recent years, spearheaded by the study of the out-of-time-ordered correlator (OTOC) and its interplay with information scrambling~\cite{larkin_quasiclassical_1969,kitaev_simple_2015,MaldacenaChaos2016,PhysRevLett.115.131603,PolchinskiSYK2016,MezeiChaos2017,Roberts2017Chaos}. The precise role that the OTOC plays in characterizing quantum chaos, via its short-time exponential growth, is well-understood in systems with either (i) a semiclassical limit, or (ii) with a large number of local degrees of freedom~\cite{kitaev_simple_2015,MaldacenaChaos2016}. 

However, its role in finite systems, such as quantum spin-chains is still under close examination~\cite{PhysRevB.98.134303,PhysRevLett.123.160401,luitz2017information,PhysRevE.101.010202,PhysRevLett.124.140602,hashimoto2020exponential}; see also Ref.~\cite{wang2020quantum} debating some of these results. OTOCs have also been applied to study a variety of many-body phenomena, ranging from quantum phase transitions~\cite{PhysRevLett.123.140602} all the way to many-body localization~\cite{Huang_2016_mbl,2017FanOTOCMBL,chen2016universal,Chen_2016,He_2017,Swingle_2017}. Recently, a connection between OTOCs, coherence-generating power, and geometry was unveiled in Ref.~\cite{anand2020quantum}. This further qualifies the intuition that the OTOC measures incompatibility between observables~\cite{Yunger_Halpern_2019}. Moreover, in Refs.~\cite{leone2020isospectral,oliviero2020random} various quantifiers of chaos were unified under the framework of \textit{isospectral twirling}. The OTOCs' theoretical investigations have also been complemented with several state-of-the-art experiments, where dynamical features of the OTOC were studied using superconducting qubits ~\cite{mi2021information,braumuller2021probing}, nuclear magnetic resonance~\cite{Wei_2018,PhysRevX.7.031011,nie2019detecting,Nie_2020}, ion-trap quantum simulators~\cite{G_rttner_2017,Joshi_2020}, among others~\cite{Meier_2019,Chen_2020}.

In recent works it was noted that, for various finite-dimensional many-body systems with spatial locality, the \textit{equilibration value} of OTOCs can diagnose the chaotic-vs-integrable nature of dynamics~\cite{PhysRevLett.121.210601,PhysRevE.100.042201,styliaris_information_2020}. In particular, this emphasis on locality was essential in establishing the connection~\cite{yan2020information} between OTOCs and Loschmidt Echo~\cite{peres1984stability,PhysRevLett.86.2490,Goussev:2012,Gorin2006LEreview}, a well-established signature of quantum chaos. Many qualitative features of the OTOC are insensitive to the specific choice of operators, as long as their locality is fixed. Therefore, it constitutes a meaningful simplification to focus on OTOCs averaged over (suitably distributed) random operators.

Given a bipartition of the system Hilbert space, one can analytically perform the uniform average over pairs of random unitary operators, supported over either side of the bipartition~\cite{styliaris_information_2020}. This averaged bipartite OTOC has a two-fold operational significance: (i) it quantifies the operator entanglement of the dynamics~\cite{zanardi2001entanglement,wang2002entanglement}, and (ii) it  quantifies average entropy production as well the  scrambling of information at the level of quantum channels. 

Moreover, the equilibration value of the OTOCs was shown to be sensitive to the amount of structure in the spectrum (for e.g., quasi-free versus nonintegrable models have degenerate versus generic spectrum, respectively). This induces a hierarchy of constraints that can be utilized to bound the OTOC's equilibration value. Remarkably, the equilibration value of the OTOC also contains information about the
entanglement of the \emph{full} system of Hamiltonian eigenstates~\cite{styliaris_information_2020}. 
Note that, averaging the OTOC over local, random operators, supported on a bipartition was also studied in Refs.~\cite{hosur2016chaos,fan2017out}.

All the above provides compelling evidence that the averaged bipartite OTOC is a powerful tool to investigate information scrambling and chaos in many-body quantum systems.  In this paper, we will extend this formalism to open quantum systems, i.e., systems coupled to an environment, which undergo a \textit{non-unitary} time evolution. In fact, these are the systems that are directly relevant to experimental situations~\cite{Landsman_2019,blok2021quantum} and to current, as well as future-technologies for quantum information processing~\cite{PhysRevX.7.031011,Landsman_2019,G_rttner_2017}.

We note that open-system effects in information scrambling have also been reported  before in Refs.~\cite{PhysRevA.97.062113,PhysRevB.99.014303,PhysRevX.9.011006,PhysRevLett.122.040404,dominguez2020decoherence,xu2020thermofield,PhysRevLett.122.014103,touil2020information,PhysRevB.97.161114}. 
However, our focus is on the open-system version of the \textit{bipartite} averaged OTOC, which, as mentioned before, has a clear operational content~\cite{styliaris_information_2020}.

The paper is structured as follows. In~\cref{sec:Preliminaries}, we discuss the general results extending to the open system domain, those of Ref.~\cite{styliaris_information_2020}. In~\cref{sec:special-channels}, we analyze a few relevant examples of quantum channels amenable of full analytical treatment, e.g., random dephasing. In~\cref{sec:spin-chains}, we discuss, with the help of numerical means, the application of our formalism to paradigmatic dissipative quantum spin-chains featuring regular and chaotic behavior. In~\cref{sec:conclusions}, we conclude with a brief discussion of our results. The detailed proofs of our main propositions are collected in the~\cref{sec:appendix}.

\section{General  results}
\label{sec:Preliminaries}

Let \(\mathcal{H} \cong \mathbb{C}^{d}\) be the Hilbert space corresponding to a \(d\)-dimensional quantum system with \(\mathcal{L}(\mathcal{H})\) denoting the space of  linear operators on \(\mathcal{H}\). Quantum states are represented by \(\rho \in \mathcal{L}(\mathcal{H})\), such that \(\rho \geq 0\) and \(\operatorname{Tr} \rho =1\).  The  space \(\mathcal{L}(\mathcal{H})\) can be endowed with a Hilbert-Schmidt inner product \(\left\langle X,Y \right\rangle := \operatorname{Tr}\left[ X^{\dagger} Y \right]\), transforming it into a Hilbert space. 

\subsection{Preliminaries}
The evolution of quantum states is described via \textit{quantum channels}, linear superoperators \(\mathcal{E}: \mathcal{L}(\mathcal{H}) \rightarrow \mathcal{L}(\mathcal{K})\) that are completely positive and trace preserving (CPTP). The time evolution of observables is via the adjoint channel, \(\mathcal{E}^{\dagger}\) which is defined as,
\begin{align}
\left\langle X, \mathcal{E}(Y) \right\rangle = \left\langle \mathcal{E}^{\dagger}(X),Y \right\rangle ~~\forall X \in \mathcal{L}(\mathcal{H}), Y \in \mathcal{L}(\mathcal{K}) .
\end{align}

For closed quantum systems, the dynamics is described by a family of unitary channels, \(\mathcal{U}_{t}(X) := U^{\dagger}_{t} X U_{t}\), where $U_t\in\mathcal{U}(\mathcal{H})$ ($=$ unitary group over the Hilbert space \(\mathcal{H}\)) $\forall t.$

Given a unitary dynamics $\{U_t\}_{t\geq 0}$ over $\mathcal{H}$, the fundamental quantity that we will use  to quantify information scrambling is given by the ``the square of the commutator'' between an operator $W$ and a time-evolved one \(V(t) := U_{t}^{\dagger} V U_{t}\),
\begin{align}
C_{V, W}(t):=\frac{1}{2d} \|[V(t), W]\|_2^2,
\end{align}
where $\|X\|_2:=\sqrt{\langle X, X\rangle}$. If we choose \(V,W\) to be unitary, then the commutator \(C_{V,W}(t)\) is related to the four-point correlation function,
\begin{align}
F_{V,W}(t) := \frac{1}{d} \operatorname{Tr}\left(V^{\dagger}(t) W^{\dagger} V(t) W \right),
\end{align}
as 
\begin{align}
C_{V,W}(t) = 1 - \frac{1}{d}\mathrm{Re} F_{V,W}(t).
\end{align}
The four-point function $F_{V,W}(t)$ with unusual time-ordering is the so-called called the ``out-of-time-ordered correlator'' (OTOC). Note that, we will be working with the infinite-temperature case throughout this paper, hence the factor of $1/d$ in the OTOC (and the associated squared commutator).

Following~\cite{styliaris_information_2020}, we will from now on consider a bipartite Hilbert space, \(\mathcal{H}_{AB}=\mathcal{H}_{A} \otimes \mathcal{H}_{B} \cong \mathbb{C}^{d_{A}} \otimes \mathbb{C}^{d_{B}}\) and define the \emph{averaged bipartite  OTOCs} by
\begin{align}\label{eq:def-botoc}
  G(\mathcal{U}_{t}):= \mathbb{E}_{V_A,W_B} \left[C_{V_A, W_B}(t)\right],
\end{align}
where, \(V_{A} = V \otimes I_{B}, W_{B} = I_{A} \otimes W\), with \( V\in\mathcal{U}(\mathcal{H}_A),\,W\in\mathcal{U}(\mathcal{H}_B),\) and $\mathbb{E}_{V,W}\left[\bullet\right]:= \int_{\mathrm{Haar}} dV\, dW \left[\bullet\right] $ denotes Haar-averaging over the standard uniform measure over $\mathcal{U}(\mathcal{H}_{A(B)})$. We emphasize that, in this work (and Ref.~\cite{styliaris_information_2020}), the Haar-averages are performed over the operators \(V,W\) in the OTOC but \textit{not} over the dynamical unitary \(U_{t}\), which is left as an input to this correlation function.~\cref{eq:def-botoc} defines the key quantity of this paper. In Ref.~\cite{styliaris_information_2020} we showed that the double-average in~\cref{eq:def-botoc} can be performed analytically and for unitary dynamics, the averaged bipartite OTOC takes the following form. Throughout this paper, we will use primed subsystems $A'$ to refer to a replica of a subsystem $A$, i.e., $\mathcal{H}_{A} \cong \mathcal{H}_{A'}$, $\mathcal{H}_{B} \cong \mathcal{H}_{B'}$, and so on.\\
\begin{prop}
\label{th:botoc-op-ent-closed}\cite{styliaris_information_2020}
Let \(S_{A A^{\prime}}\) be the operator over \(\mathcal{H}_{AB} \otimes \mathcal{H}_{A'B'}\) that swaps \(A\) with its replica \(A^{\prime}\), one has 
\begin{align}
\label{eq:botoc-op-ent-closed}
  G(\mathcal{U}_{t})=1-\frac{1}{d^{2}} \operatorname{Tr}\left(S_{A A^{\prime}} U_{t}^{\otimes 2} S_{A A^{\prime}} U_{t}^{\dagger \otimes 2}\right).
\end{align}
\end{prop}
This simple formula --- which, quite surprisingly, \emph{coincides}  with the operator entanglement of $U_t$ as originally defined in Ref.~\cite{zanardi2001entanglement} --- provides the starting point of
the analysis in~\cite{styliaris_information_2020}. It allows one to connect the averaged bipartite OTOC to a variety of physical and information-theoretic quantities e.g., entropy production, channel distinguishability, among others. For completeness, we review some of these ideas in~\cref{sec:operator-entanglement-intro}.

We are now ready to discuss the generalization of the bipartite OTOC formalism to open quantum systems, where, unitary transformations are replaced by more general quantum operations.

\subsection{Open OTOC}
Assuming that standard Markovian properties hold, the system dynamics in the Schr\"odinger picture is then described by a trace-preserving, completely positive (CP) map, also known as a quantum channel ${\mathcal E}^\dagger$~\cite{breuerTheoryOpenQuantum2002}. It follows that in the Heisenberg picture (i.e., the one adopted throughout this paper), the observable dynamics is described by the \emph{unital} CP map $\mathcal E$. Recall that a quantum channel \(\mathcal{E}\) is called unital if and only if \(\mathcal{E}(\frac{\mathbb{I}}{d}) = \frac{\mathbb{I}}{d}\), where \(\frac{\mathbb{I}}{d}\) is the maximally mixed state (or the Gibbs state at infinite temperature). Namely, such a map has the maximally mixed state as a fixed point. Several important physical operations that one can perform on a quantum system are unital, for example, unitary evolution, projective measurements without post-selection, dephasing channels, among others. A quantum channel $\mathcal{E}$ is trace preserving if and only if ${\mathcal E}^\dagger$ is itself unital. While many of the results and ideas which follow do not rely on this  assumption, for the sake of simplicity, we will assume
that ${\mathcal E}^\dagger$ is indeed unital ($\Rightarrow \mathcal{E}$ is a quantum channel).


We define the open (averaged) bipartite OTOC by, 
\begin{align}
  G(\mathcal{E}) := \frac{1}{2  d }\, \mathbb{E}_{V_A, W_B} \left\Vert \left[ \mathcal{E}(V_A),W_B \right]  \right\Vert_{2}^{2}, 
\end{align}
where $V_A$, $W_B$ and the average are as defined in~\cref{eq:def-botoc}. The first step is to  generalize~\cref{eq:botoc-op-ent-closed} to the open case.\\
\begin{restatable}{prop}{haaravgopenotoc}
\label{th:botoc-op-ent-open}
Let \(S \equiv S_{AA'BB'}\) be the swap operator over \(\mathcal{H}_{AB} \otimes \mathcal{H}_{A'B'}\), then for a quantum channel $\mathcal{E}: \mathcal{L}(\mathcal{H}_{AB}) \rightarrow \mathcal{L}(\mathcal{H}_{AB})$, the open bipartite OTOC takes the following form,
\begin{align}
\label{eq:botoc-op-ent-open}
  G(\mathcal{E}) = \frac{1}{d^{2}} \operatorname{Tr}\left( \left( d_{B} S - S_{AA'} \right) \mathcal{E}^{\otimes 2} (S_{AA'}) \right).
\end{align}
\end{restatable}
%
A few  remarks are in order:  

(a) If $L_S(X):=SX$ \footnote{Here, \(L_{S}\) is a superoperator whose action is to left multiply with the swap operator \(S\), that is, \(L_{S}(X):= SX\). The commutator is at the level of superoperators, namely, \(\left[ \mathcal{E}^{\otimes 2},L_{S} \right] = \mathcal{E}^{\otimes 2} \circ L_{S} - L_{S} \circ \mathcal{E}^{\otimes 2}\), where we have emphasized the superoperator composition via the \(\circ\) symbol. This commutator can be understood by its action on an operator \(X\) as \(\left[ \mathcal{E}^{\otimes 2},L_{S} \right](X) = \mathcal{E}^{\otimes 2} L_{S}(X) - L_{S} \mathcal{E}^{\otimes 2}(X)\).} one has that $\left[ \mathcal{E}^{\otimes 2}, L_S \right] = 0 $, if and only if $\mathcal{E}$ is unitary (see the~\cref{sec:appendix} for a proof). In this case  the first term in~\cref{eq:botoc-op-ent-open} becomes equal to one, giving back ~\cref{eq:botoc-op-ent-closed}.   

(b) From $\left[ \mathcal{E}^{\otimes 2}, L_S \right] = 0$ and $S_{BB'}=S_{AA'} S=S S_{AA'}$, one sees that the second term in~\cref{eq:botoc-op-ent-open}
can be written $S_{BB'}\mathcal{E}^{\otimes 2}(S_{BB'}).$ This means that in the unitary case there a symmetry between the subsystems $A$ and $B$ which is lost in the general open case.

(c) Since, for  unitary dynamics,~\cref{eq:botoc-op-ent-closed} coincides with operator entanglement~\cite{zanardi2001entanglement} of $U$, one has that
\begin{align}
 G(\mathcal{U}) = 0 \iff {U} = {U}_{A} \otimes {U}_{B}.
\end{align}
However, for non-unitary dynamics, \(\mathcal{E} = \mathcal{E}_{A} \otimes \mathcal{E}_{B} \implies G(\mathcal{E}) = 0\), but the converse is not true. Namely, 
one  can have zero \(G(\mathcal{E})\) even for \(\mathcal{E} \neq \mathcal{E}_{A} \otimes \mathcal{E}_{B}\).  Later, we will illustrate this phenomenon by an example
of a dephasing channel.


(d) Let us remind that given the quantum channel $\mathcal{E}\colon \mathcal{L}(\mathcal{H})\rightarrow \mathcal{L}(\mathcal{K})$, one defines the \emph{Choi state} associated to it by,
\begin{align}
\rho_{\mathcal{E}}:=(\mathcal{E}\otimes I)(|\Phi^+\rangle\langle\Phi^+|)\in \mathcal{L}(\mathcal{K})\otimes \mathcal{L}(\mathcal{H}),
\end{align}
where $|\Phi^+\rangle= d^{-1/2}\sum_{i=1}^d |i\rangle^{\otimes 2}\in \mathcal{L}(\mathcal{H})^{\otimes 2} ,\,(d=\mathrm{dim} \mathcal{H}).$

Notice that in the unitary case, $\mathcal{E}=\mathcal {U}= U\bullet U^\dagger,$~\cref{eq:botoc-op-ent-closed} can be written as~\cite{zanardi2001entanglement}
\begin{align}\label{eq:Choi-closed}
G(\mathcal{U}_{t})=1-\|\mathrm{tr}_{BB'} \rho_{\mathcal{U}_t}\|_2^2=S_L( \mathrm{tr}_{BB'}\rho_{\mathcal{U}_t}),
\end{align}
 where $S_L$  is the so-called linear entropy i.e., $S_L(\rho):=1 -\mathrm{Tr}(\rho^2).$
This shows why the averaged bipartite OTOC corresponds to a measure of operator entanglement for $U_t$ across the $A:B$ bipartition~\cite{zanardi2001entanglement}.

The following result can be seen as an extension of~\cref{eq:Choi-closed} to  general quantum channels.\\
\begin{restatable}{prop}{choiopen}
\label{th:Choi-open}
%
\begin{align}
\label{eq:Choi-open-0}
(i)\quad G(\mathcal{E}) = d_{B} \left\Vert \operatorname{Tr}_{B'}\rho_\mathcal{E} \right\Vert_{2}^{2} - \left\Vert \operatorname{Tr}_{BB'} \rho_\mathcal{E} \right\Vert_{2}^{2} .
\end{align}
\begin{align}
\label{eq:Choi-open}
(ii)\quad  G(\mathcal{E}) = d_{B} \left\Vert \rho_{\widetilde{\mathcal{E}}} - \rho_{\mathcal{T} \circ \widetilde{\mathcal{E}}} \right\Vert_{2}^{2}=d_{B} \left( \left\Vert \rho_{\widetilde{\mathcal{E}}} \right\Vert_{2}^{2} - \left\Vert \rho_{\mathcal{T} \circ \widetilde{\mathcal{E}}} \right\Vert_{2}^{2}   \right), 
\end{align}
where  $ \widetilde{\mathcal{E}}: \mathcal{L}(\mathcal{H}_A)\rightarrow \mathcal{L}(\mathcal{H}_{AB}): X\rightarrow \mathcal{E}(X\otimes\frac{\mathbb{I}}{d_{B}})$ and
\(\mathcal{T}: \mathcal{L}(\mathcal{H}_{AB}) \rightarrow \mathcal{L}(\mathcal{H}_{AB}): X{\mapsto} \operatorname{Tr}_{B}\left( X \right)  \otimes \frac{\mathbb{I}}{d_{B}}\).
\end{restatable}
In words: the averaged  bipartite  OTOC (\ref{eq:botoc-op-ent-open})  for a channel \(\mathcal{E}\) can be expressed as a difference of purities of (reduced) Choi matrices
of $\mathcal{E}$ or as a (squared) distance  between the Choi matrices of channels $\widetilde{\mathcal{E}}$ and ${\mathcal{T}}\circ \widetilde{\mathcal{E}}$. More precisely, since the map between channels and the corresponding Choi state is injective, the RHS of~\cref{eq:Choi-open} measures the distance  between the channels $ \widetilde{\mathcal{E}}$ and ${\mathcal{T}}\circ \widetilde{\mathcal{E}}.$
Hence, we see that $G(\mathcal{E}) = 0$ if and only if $\mathcal{\tilde{E}} = \mathcal{T} \circ \mathcal{\tilde{E}}.$ Namely, $ \forall X \in \mathcal{L}(\mathcal{H}_{A}),$ 
  \begin{align}
\mathcal{E} ( X \otimes \frac{\mathbb{I}}{d_{B}} ) = \operatorname{Tr}_{B} \mathcal{E} ( X \otimes \frac{\mathbb{I}}{d_{B}} ) \otimes \frac{\mathbb{I}}{d_{B}}.
\end{align}
In passing, we observe that the map $\mathcal{T}$ is a (super) projection that can be realized as a group average $\mathcal{T}(X)=\mathbb{E}_U \left[  (\mathbb{I}_A\otimes U)X (\mathbb{I}_A\otimes U^\dagger)\right],$ with $U\in\mathcal{U}({\cal H}_B)$.

 
From the physical point of view one of the main findings in~\cite{styliaris_information_2020}  was to  show that the bipartite OTOC $G(\mathcal{U}_{t})$ is nothing but a measure of the average entropy production by \(\operatorname{Tr}_{B}[\widetilde{\mathcal{E}}]\) over pure states. Operationally, one prepares pure states in the $A$-subsystem tensorized with the totally mixed one in the $B$-subsystem and lets the joint system evolve according to the channel $\mathcal{E}.$ The entropy that is then observed in the $A$-subsystem alone is the result, in the unitary case, of the information loss due to leaking into the $B$-subsystem induced by the evolution  i.e., quantum information scrambling.

One can  extend  that key result to the  open system case.\\
\begin{restatable}{prop}{entropyopen}
\label{th:entropy-production-open}
We denote by $\psi:=|\psi\rangle\langle\psi|$ with $|\psi\rangle\in\mathcal{H}_A$. Then,
\begin{align}
\label{eq:entropy-production-open}
  G(\mathcal{E}) = N_{A} \mathbb{E}_{\psi} \left[ S_{L} ( \operatorname{Tr}_{B} \widetilde{\mathcal{E}} \left( \psi \right) )  \right.   
   \left.   - d_{B} (S_{L} ( \widetilde{\mathcal{E}}(\psi)) - S_{L}^{\mathrm{min}} )\right],
\end{align}
where $ \mathbb{E}_{\psi} $   is the the Haar average over $\mathcal{H}_A,$ $N_{A}:=\frac{d_{A}+1}{d_{A}}$, and $S_{L}^{\mathrm{min}} := 1 - \frac{1}{d_{B}}.$
\end{restatable}
We note that for \(\mathcal{E} = \mathcal{U}\), that is, closed system dynamics, the second term in~\cref{eq:entropy-production-open} is zero. 
In general, since $\mathcal{E}$ is unital, one has that $ \widetilde{\mathcal{E}}(\psi)) \ge S_{L}^{\mathrm{m}}\,(\forall\psi).$ Hence, 
\begin{align}\label{eq:upper-bound}
G(\mathcal{E})\le G^{\mathrm{scra}}(\mathcal{E})\le G^{\mathrm{max}}:= 1-\frac{1}{d_A^2},
\end{align}
where, the ``scrambling entropy'' production $ G^{\mathrm{scra}}$ is given by the first term in~\cref{eq:entropy-production-open}.
%

Crucially, 
\cref{eq:entropy-production-open} shows that in the open system case in \(G(\mathcal{E})\), there is a competition between the entropy production, quantified by  the first term $G^{\mathrm{scra}}$ due to scrambling, and the second one due to decoherence.
%
For example, if $\widetilde{\mathcal{E}} \left( \psi \right) =\frac{\mathbb{I}}{d},\,\forall\psi$, then, the scrambling term attains its maximum value $G_{\mathrm{max}}$, but this is {\em{exactly canceled}}  by the decoherence contribution. This situation, as shown in the next section, can be physically realized by a dephasing channel in the maximally entangled basis.

%
%

We stress that to obtain a satisfactory estimate of the average in the RHS of Eq.~\eqref{eq:entropy-production-open}, one does not, in practice, need to sample over the full Haar ensemble. An adequate estimate can be obtained with a rapidly decreasing number of necessary samples, as the dimension $d_{\chi}$ grows. For example, in the unitary case,  if  $\tilde P(\epsilon)$ is the probability of the entropy  $S_{\mathrm{lin}} \big[ {\cal E}_t  \big( \ket{\psi}\! \bra{\psi} \big) \big]$ deviating from $\frac{d_{A} }{d_{A}+1} G(\mathcal{U}_{t})$ more than $\epsilon$ for an instance of a random state, one has that ~\cite{styliaris_information_2020}:
\begin{align} \label{eq:concentration_linear_entropy}
\tilde P(\epsilon) \le \exp \left(  - \frac{d_A \epsilon^2}{64} \right).
\end{align}
It is also important to notice that the two terms in~\cref{eq:entropy-production-open} can be, in principle, measured independently and therefore
have a well-defined operational meaning in their own right see \cref{sec:experimental-openotoc} for a detailed discussion.
\vskip .5truecm
\section{Some special channels}
\label{sec:special-channels}
\begin{figure*}[!t]
\begin{center}
\includegraphics[width=0.9\columnwidth]{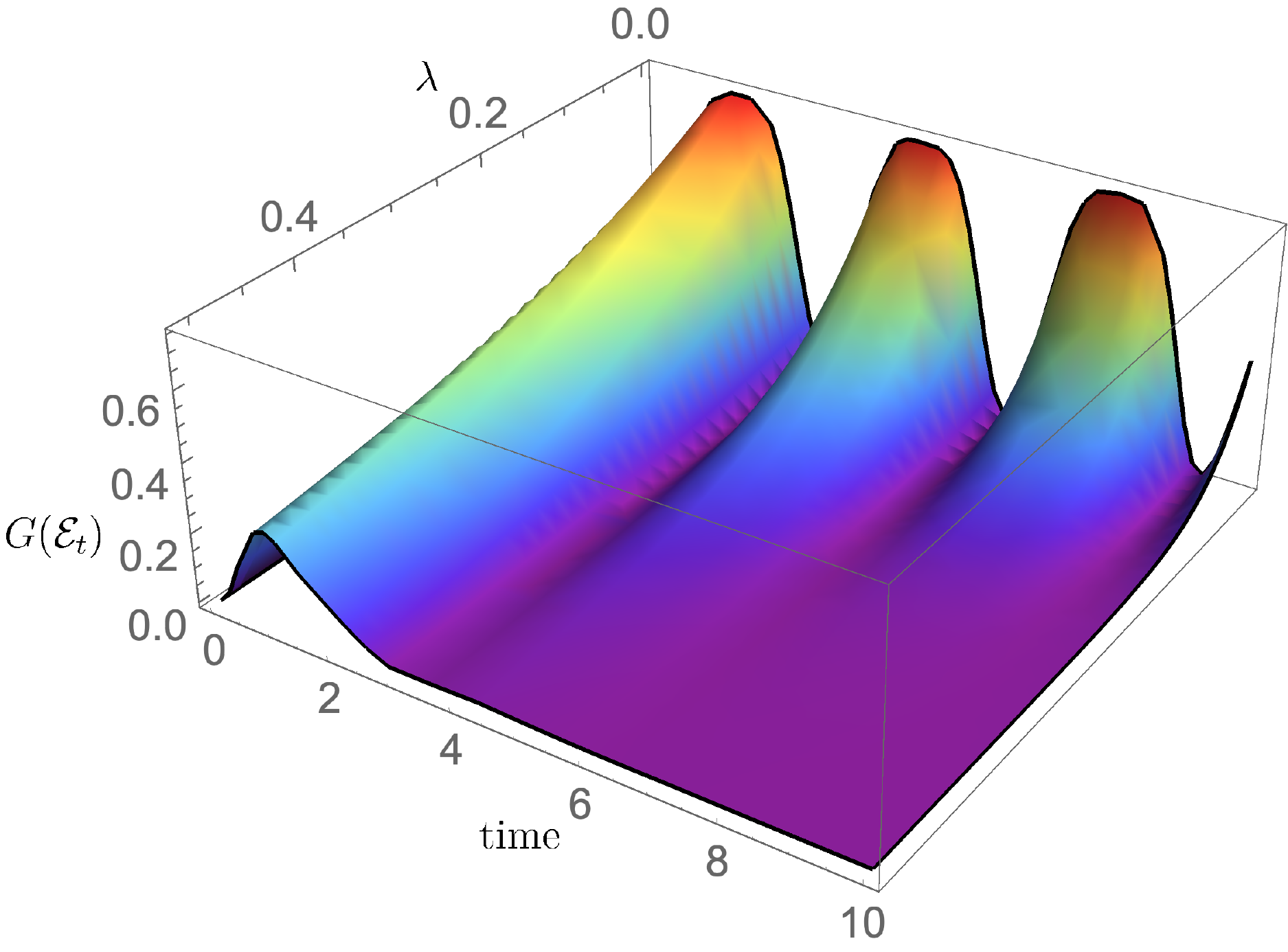}
\end{center}
\caption{\label{fig:swap-w-dissi} Non-unitary OTOC $G(e^{{\cal L}t})$ with \(\mathcal{L} = i \mathrm{ad} {S} + \lambda \left( \mathcal{D}_{\mathbb{B}} - I \right)\) where $S$ is the swap operator, \({\mathbb{B}}\)  is  the Bell basis ($d_A=d_B=2$). The different curves correspond to different choices of the dephasing parameter $\lambda.$ Over the time scale $\lambda^{-1}$ on which dephasing becomes relevant the ``scrambling entropy'' (first term in~\cref{eq:entropy-production-open}) is balanced, and eventually overwhelmed, by the  decoherence-induced entropy production (second term in~\cref{eq:entropy-production-open}). For any fixed time $t$ the OTOC suppression is exponential in the dephasing strength $\lambda.$ 
Moreover, in sharp contrast with the unitary case, for any $\lambda\neq 0,$ the infinite time limit of the OTOC is vanishing.}
\end{figure*}

For concreteness, let us now consider a family of  maps which includes several ones of physical interest  and for which~\cref{eq:botoc-op-ent-open} takes a particularly interesting form. Let us start with dephasing channels \emph{stricto sensu}.
\subsection{Dephasing channels}
\begin{restatable}{prop}{dephasingchannel}
\label{th:dephasing-channel}
Consider the dephasing channel, \(\mathcal{E} = \mathcal{D}_{\mathbb{B}}\), where, \(\mathcal{D}_{\mathbb{B}}(\rho) = \sum\limits_{\alpha=1}^{d} \Pi_{\alpha} \rho \Pi_{\alpha}\) and \(\mathbb{B} = \{ \Pi_{\alpha} \}_{\alpha=1}^{d}\) with \(\Pi_{\alpha} = | \psi_{\alpha} \rangle \langle  \psi_{\alpha} |\), an orthonormal basis. Then,
\begin{align}\label{eq:idempotency-deficit} 
  G(\mathcal{D}_{\mathbb{B}}) = 
  \frac{1}{d^{2}_{A}} \left\Vert \widetilde{X}_{\mathbb{B}} - \widetilde{X}_{\mathbb{B}}^{2} \right\Vert_{1},
\end{align}
where \( ( \widetilde{X}_{\mathbb{B}} )_{\alpha,\beta} : = d_{B}^{-1}\left\langle \rho_{\alpha}, \rho_{\beta} \right\rangle\) is the renormalized Gram matrix of the system and \(\rho_{\alpha} = \operatorname{Tr}_{B}\left( \Pi_{\alpha} \right) ~~\forall \alpha\). 
\end{restatable}
%
\cref{eq:idempotency-deficit} describes an ``idempotency deficit,''  namely how far away \(\widetilde{X}_{\mathbb{B}}\) is from being equal to its own square \(\widetilde{X}_{\mathbb{B}}^{2}\). 
Hence,
 $ G(\mathcal{D}_{\mathbb{B}}) = 0$ if and only if $\widetilde{X}_{\mathbb{B}} \text{ is a projector.}$

Define $|\phi^{s} \rangle_{X}:= \frac{1}{\sqrt{d_X}}\sum_{j=1}^{d_X} |j\rangle,\,(X=A,B)$ and consider the following two examples of vanishing $G$.
\begin{enumerate}
\item[(i)] A product dephasing channel, i.e., \(\mathcal{D}_{\mathbb{B}} = \mathcal{D}_{\mathbb{B}_{A}} \otimes \mathcal{D}_{\mathbb{B}_{A}}\). Let \(\mathbb{B}_{A} = \{ P_{j} \}_{j=1}^{d_{A}}\) and \(\mathbb{B}_{B} = \{ Q_{j} \}_{j=1}^{d_{B}}\), then the projectors corresponding to \(\mathcal{D}_{\mathbb{B}}\) are \(\{ P_{j} \otimes Q_{k} \}_{j,k=1}^{d_{A},d_{B}}\). It is easy to show that the Gram matrix corresponding to \(\mathbb{B}\) takes the form, \(\widetilde{X} = \mathbb{I}_{A} \otimes | \phi^{s} \rangle_{B} \langle  \phi^{s} |\).  
 
\item[(ii)] A maximally entangled dephasing basis, i.e., each of the \(\rho_{\alpha} = \mathbb{I}_{A}/d_{A}\) and therefore, a simple calculation shows that the Gram matrix takes the form, \(\widetilde{X} =| \phi^{s} \rangle_{A} \langle  \phi^{s} |\otimes | \phi^{s} \rangle_{B} \langle  \phi^{s} |\).
\end{enumerate}


Quite interestingly,~\cref{eq:idempotency-deficit}  allows one to connect $G(\mathcal{D}_{\mathbb{B}})$ to the entanglement of the states comprising $\mathbb{B}$.\\
\begin{restatable}{prop}{deficitentanglement}
\label{th:deficit-entanglement} Let \(\mathbb{B} = \{ \Pi_{\alpha} \}\), \(\rho_{\alpha} = \operatorname{Tr}_{B}\left[ \Pi_{\alpha} \right]\), and, $\Delta_\alpha:=\rho_\alpha-\frac{\mathbb{I}}{d_A}$. Then if $\|\Delta_\alpha\|_2^2\le\varepsilon\,(\forall\alpha)$, one has the following upper bound on the open OTOC for dephasing channels, 
$G(\mathcal{D}_{\mathbb{B}})\le\frac{\varepsilon}{d_A}$.
\end{restatable}

Since $\Pi_{\alpha}$ are pure states, if $\Delta_\alpha$ is small, then the states \(\Pi_{\alpha}\) are nearly maximally entangled across the \(A:B\) partition. Therefore, the bound then tells us that, \textit{the more entangled the dephasing basis states, the smaller the OTOC}. Note that the assumption above, in order to make its connection to entanglement clearer, can also be recast as $$S_L(\rho_\alpha)\ge S_L^{\mathrm{max}}-\varepsilon\quad(\alpha=1,\ldots,d),$$
where $S_L^{\mathrm{max}}:=1-1/d_A.$ 

Another useful  way of rewriting~\cref{eq:idempotency-deficit} is obtained by introducing the following $\mathbb{B}$-dependent state, $R_{\mathbb{B}} \in {\cal H}^{\otimes 2}\cong  {\cal H}_A\otimes{\cal H}_B\otimes{\cal H}_{A'}\otimes{\cal H}_{B'}$ such that
\begin{align}\label{eq:R-matrix}
R_{\mathbb{B}}:=\frac{1}{d}\sum_{\alpha=1}^d \Pi_\alpha\otimes\Pi_\alpha=(\mathcal{D}_{\mathbb{B}}\otimes I)(|\Phi^+_{AB}\rangle\langle\Phi^+_{AB}| ),
\end{align}
where $|\Phi^+_{AB}\rangle:=d^{-1/2}\sum_{\alpha=1}^d|\psi_\alpha\rangle^{\otimes 2}.$ The second equality above shows that $R_{\mathbb{B}}$ is nothing but the Choi state associated to $\mathcal{D}_{\mathbb{B}}.$
Using Eqs. (\ref{eq:Choi-open-0}) (or (\ref{eq:idempotency-deficit})) and  (\ref{eq:R-matrix})  one can write
\begin{align}\label{eq:G_B-R-matrix}
G(\mathcal{D}_{\mathbb{B}})=\frac{1}{d_A}\langle S_{AA'}, R_{\mathbb{B}}\rangle -\|R_{\mathbb{B}}^{AA'}\|_2^2,
\end{align}
where $R_{\mathbb{B}}^{AA'}:=\mathrm{Tr}_{BB'} R_{\mathbb{B}}.$ Since the first term in~\cref{eq:G_B-R-matrix} is upper bounded by $1$ and the second term is lower bounded by $d_A^2$, one immediately obtains the $\mathbb{B}$-independent upper bound 
\begin{align}\label{eq:G_B-upper-bound}
G(\mathcal{D}_{\mathbb{B}})\le\frac{1}{d_A}\left(1-\frac{1}{d_A}\right)=O(\frac{1}{d_A}).
\end{align}
This inequality shows that the maximal value of the OTOC that is achievable by dephasing channels is well below the upper bound~\cref{eq:upper-bound}, $G^{\mathrm{max}}=1-1/d_A^{2}$.

To explore this phenomenon we now move to consider \emph{random} dephasing channels. The set of $\mathbb{B}$'s is naturally acted upon by the unitary group $\mathcal{U}(\mathcal{H})$ \footnote{The key idea is that any two bases in the Hilbert space can be connected via a unitary. Therefore, starting from a fixed basis \(\mathbb{B}_{0}\), the action of the unitary group generates \textit{all bases} in the Hilbert space. Then, utilizing the uniform (Haar) measure on $\mathcal{U}(\mathcal{H})$ allows us to define a notion of (uniformly distributed) random bases.}:
$$\mathbb{B}_0:=\{\Pi_\alpha^{(0)}\}_{\alpha=1}^d\mapsto U\cdot\mathbb{B}_0:=\{U\Pi_\alpha^{(0)}U^\dagger\}_{\alpha=1}^d.$$
In terms of the $R_{\mathbb{B}}$ matrices: $R_{\mathbb{B}_0}\mapsto U^{\otimes 2}  R_{\mathbb{B}_0}U^{\dagger\otimes 2}.$ By considering the $U$'s Haar distributed one obtains the desired ensemble of random dephasing channels.  The next proposition shows the average and measure concentration for $G(\mathcal{D}_{\mathbb{B}})$ for such an ensemble with $d_A\le d_B$.\\
\begin{restatable}{prop}{randomdephasingchannels}
\label{th:random-dephasing-channels} 
i) $\mathbb{E}_{U}\left[ G(\mathcal{D}_{U\cdot\mathbb{B}_0)}\right]\le \frac{7}{4\, d_A^2}=O(\frac{1}{d_A^2}).$

ii) $
\mathrm{Prob} \{ G(\mathcal{D}_{\mathbb{B}}) \geq \frac{7}{4 d_{A}^{2}} + \epsilon \} \leq \exp \left[ -d \epsilon^{2}/K^{2}\right], 
$
where $K$ is the Lipschitz constant of the function $F(U):=  G(\mathcal{D}_{U\cdot\mathbb{B}_0})$ and can be chosen     $K \ge 100.$
\end{restatable}
In words: in large dimension the overwhelming majority of random dephasing channels have a $G(\mathcal{D}_{\mathbb{B}})$ which is $(1/d^2_A).$
This is the result of decoherence which makes the first term in~\cref{eq:G_B-R-matrix} (or~\cref{eq:Choi-open-0})  being $O(1/d^2_A)$ for typical dephasing channels. On the other hand, such a term in the closed case 
is identically one and typical unitaries have a $G(\mathcal{U})$ which is close to  $G^{\mathrm{max}}$~\cite{styliaris_information_2020}.
\subsection{Entanglement-breaking channels}
Here we discuss the class of channels called entanglement-breaking or measure-and-prepare, defined (in the Heisenberg picture) as,
\begin{align}\label{eq:EB-channel}
\Phi_{\mathrm{EB}}(X) = \sum\limits_{k}^{} M_{k} \operatorname{Tr}\left[ \delta_{k} X \right], \text{ where } \sum\limits_{k}^{} M_{k} = \mathbb{I}.
\end{align}
Here, \(\Phi_{\mathrm{EB}}: \mathcal{L}(\mathcal{H}_{AB}) \rightarrow \mathcal{L}(\mathcal{H}_{AB})\) where \(\{ M_{k} \}, \{ \delta_{k} \}\) are linear operators on \(\mathcal{L}(\mathcal{H}_{AB})\) with the additional constraint that \(\{ M_{k} \}_{k}\) form a POVM and \(\{ \delta_{k} \}_{k}\) is a set of quantum states.

For general EB channels, we have the following form.\\
\begin{restatable}{prop}{entanglementbreakingchannel}
\label{th:entanglement-breaking-channel}
Consider a general entanglement-breaking (EB) channel as in~\cref{eq:EB-channel} then,
\begin{align}
  G(\Phi_{\mathrm{EB}}) =  \frac{1}{d^{2}} \sum\limits_{k,k'}^{}  \left\langle \delta_{k}^{A}, \delta_{k'}^{A} \right\rangle \left[ d_{B} \left\langle M_{k}, M_{k'} \right\rangle -  \left\langle M_{k}^{A}, M_{k'}^{A} \right\rangle  \right],
\end{align}
where \(M_{k}^{A} \equiv \operatorname{Tr}_{B} M_{k} \) and \(\delta_{k}^{A} \equiv \operatorname{Tr}_{B} \delta_{k} \).
\end{restatable}
Note that dephasing channels are a special case of EB channels when the measurements are rank-\(1\) projectors and the prepared states are (the same) pure states; that is, let \(\mathbb{B} = \{ | \psi_{k} \rangle \langle  \psi_{k} |  \}_{k=1}^{d}\) and \(\delta_{k} = | \psi_{k} \rangle \langle  \psi_{k} | = M_{k} ~~\forall k = \{ 1,2, \cdots,d \}\), then, \(\Phi_{\mathrm{EB}} = \mathcal{D}_{\mathbb{B}}\). Therefore, $G(\Phi_{\mathrm{EB}})$ takes the analytical form in~\cref{eq:idempotency-deficit}.

As an example, one can consider the following form of EB channel. Let \(\mathbb{B} = \{ \Pi_{\alpha} \}, \Pi_{\alpha} = | \psi_{\alpha} \rangle \langle  \psi_{\alpha} | \) and \(\widetilde{\mathbb{B}} = \{ \widetilde{\Pi_{\alpha}}   \}, \widetilde{\Pi_{\alpha}} = |\phi_{\alpha} \rangle \langle  \phi_{\alpha} |\) be two bases for \(\mathcal{H}_{AB}\). Then,
\begin{align}
\Phi_{\mathrm{EB}}^{(\mathbb{B} \rightarrow \widetilde{\mathbb{B}})}(X):= \sum\limits_{k=1}^{d} \widetilde{\Pi}_{\alpha} \langle \psi_{\alpha} | X |  \psi_{\alpha} \rangle. 
\end{align}
For this class of channels, we have the following form of the open OTOC. Let \(\rho_{\alpha}:= \operatorname{Tr}_{B} \Pi_{\alpha},\, \widetilde{\rho}_{\alpha}:= \operatorname{Tr}_{B} \widetilde{\Pi}_{\alpha}\). Then,
\begin{align}\label{eq:EB-B-B}
G(\Phi_{\mathrm{EB}}^{(\mathbb{B} \rightarrow \widetilde{\mathbb{B}})}) = \frac{1}{d^{2}} \left( d_{B} \sum\limits_{k=1}^{d} \left\Vert {\rho}_{k}\right\Vert_{2}^{2} - \sum\limits_{k,k'=1}^{d}\left\langle \rho_{k}, \rho_{k'} \right\rangle  \left\langle \widetilde{\rho}_{k}, \widetilde{\rho}_{k'} \right\rangle  \right).
\end{align}
For \(\mathbb{B} = \widetilde{\mathbb{B}}\) (with the identical ordering of states), this takes the form of the dephasing channel. 

It is easy to see that also for~\cref{eq:EB-B-B} the bound (\ref{eq:G_B-upper-bound})  holds. 
Indeed, the first term in~\cref{eq:EB-B-B} is clearly upper bounded by $1/d_A$ i.e., when all the $\rho_k$'s are pure) and the second can be written 
$\|d^{-1}\sum_k \rho_k\otimes \widetilde{\rho}_{k}\|_2^2$ and therefore is lower bounded by $1/d_A^2.$ This is achieved for $\mathbb{B}$ being a product basis and $\widetilde{\mathbb{B}}$ maximally entangled
i.e., $\widetilde{\rho}_k=\mathbb{I}/d_A, (\forall k).$

\subsection{$\mathbb{B}$-diagonal channels}
Let us now move to analyze a generalization of the above which we refer to as $\mathbb{B}$-diagonal channels. 
Consider a basis $\mathbb{B}:= \{ | \alpha \rangle \}_{\alpha=1}^d $ of $\mathcal{H}_{AB}$ and map $\mathcal{E}_{\hat{\Phi}}$ such that 
\begin{align}\label{eq:gen-dephasing-channel-0}
\mathcal{E}_{\hat{\Phi}}(|\alpha\rangle\langle\alpha'|)= \phi_{\alpha, \alpha'}|\alpha\rangle\langle\alpha'|\quad(\forall\alpha,\alpha'),
\end{align}
with \(\phi_{\alpha,\alpha'} \in \mathbb{C} ~~\forall \alpha,\alpha'\). This family, for example,  comprises unitary channels, dephasing channels and quantum measurements. We can then prove the following.\\
\begin{restatable}{prop}{generaldephasingchannel}
\label{th:gen-dephasing-channel} 
i) If  $\hat{\Phi}:=\left(\phi_{\alpha, \alpha'}\right)_{\alpha,\alpha'}\ge 0,$ 
and $\phi_{\alpha, \alpha}=1 \, (\forall \alpha),$ then~\cref{eq:gen-dephasing-channel-0} defines a (unital) quantum channel whose eigenvalues are encoded
in the matrix $\hat{\Phi}.$ 

ii)  \(\rho_{\alpha, \alpha'} = \operatorname{Tr}_{B}| \alpha \rangle \langle  \alpha' | \), then 
\begin{align}
\label{eq:gen-dephasing-channel}
G(\mathcal{E}_{\hat{\Phi}}) &= \frac{d_B}{d^{2}}  \sum\limits_{\alpha, \alpha'}^{} \left| \phi_{\alpha, \alpha'} \right|^2 \left\Vert \rho_{\alpha, \alpha'} \right\Vert_{2}^{2}  \nonumber\\
& -\frac{1}{d^{2}}\sum\limits_{\alpha, \alpha', \beta, \beta'}^{} \phi^{*}_{\alpha, \alpha'} \phi_{\beta, \beta'} \left| \left\langle \rho_{\alpha, \alpha'} , \rho_{\beta, \beta'} \right\rangle \right|^{2}.
\end{align}
\end{restatable}
We note the following facts:

(a) For \( \hat{\Phi}=\mathbf{1}\), we recover the dephasing channel \(\mathcal{D}_{\mathbb{B}}\) and~\cref{eq:gen-dephasing-channel} becomes~\cref{eq:idempotency-deficit} . 

(b) For $\phi_{\alpha, \alpha'}=e^{i(\theta_\alpha-\theta_{\alpha'})}$ with \(\{ \theta_{\alpha} \}_{\alpha} \in [0,2\pi)\), one recovers unitary channels and~\cref{eq:gen-dephasing-channel}  becomes~\cref{eq:botoc-op-ent-closed}. In particular if \(\phi_{\alpha, \alpha'} = 1 ~(\forall \alpha, \alpha')\) we have, \(\mathcal{E}_{\hat{\Phi}} = \mathcal{I}\) and therefore $G$ vanishes.

(c) Suppose that  the dynamics is generated by a Lindbladian 
$$\mathcal {L}(X)=\sum_\mu\left( L_\mu X L_\mu^\dagger -\frac{1}{2}\{ L_\mu L_\mu^\dagger, X\}\right),$$
where the Lindblad operators $L_\mu$ form an abelian algebra and $\{X,Y\}:= XY + YX$. Then one one has that $\mathcal{E}_t=e^{t\mathcal{L}}$ is of the form~\cref{eq:gen-dephasing-channel-0} with 
$$\phi_{\alpha,\alpha'}=
\exp\left[ -\frac{1}{2}\sum_\mu (|\alpha_\mu-\alpha'_\mu |^2 -2i\mathrm{Im}(\alpha'_\mu\bar{\alpha_\mu}))\right],$$
being the $|\alpha\rangle$ a joint eigenbasis of the $L_\mu$ i.e., $L_\mu|\alpha\rangle=\alpha_\mu |\alpha\rangle, L^\dagger_\mu|\alpha\rangle=\bar{\alpha_\mu} |\alpha\rangle,(\forall \mu,\alpha)$.

To illustrate the physical relevance of the family of channels in~\cref{eq:gen-dephasing-channel} we now provide a couple of simple analytical  examples arising from a  dynamical semigroup. They are aimed at making manifest non-unitary effects and their interplay with unitary ones.
For both examples below, \(\mathcal{H}_{A} \cong \mathcal{H}_{B}\)  and the relevant basis 
  $\{|\alpha\rangle\}$ is an eigenbasis of the swap operator $S$, for e.g., the Bell basis for $d_A=d_B=2.$

{\em{Example 1.--}} Let us consider the  Lindbladian  $$\mathcal{L} = \mathrm{Ad} S - I,$$ where  $\mathrm{Ad} S(X):= S XS.$ Then, by a straightforward exponentiation one finds 
a convex combination of unitaries
$$\mathcal{E}_{t} = e^{\mathcal{L} t} =a(t) I + b(t) \mathrm{Ad} S$$
with \(a(t) = \frac{1}{2}(1 + e^{-t}), \,b(t) =\frac{1}{2}(1 - e^{-t})\). The Lindbladian here is designed to generate an evolution which is a mixture of the Identity and the SWAP unitaries. The idea is that the swap unitary maximizes the (unitary) bipartite OTOC, while the Identity channel corresponds to zero bipartite OTOC. The probabilities for these two evolutions are time-dependent and, as time evolves, the weight corresponding to the swap unitary increases exponentially (from zero) while that of the Identity decays (from one) to zero. Namely, it generates a maximally scrambling evolution with increasing time. We have, $ \hat{\Phi}_{\alpha\alpha'}= a(t)+b(t) \lambda_\alpha\lambda_{\alpha'}$     with  $\lambda_{\alpha/\alpha'}=\pm 1.$
 The open averaged bipartite OTOC for this channel is $$G(\mathcal{E}_t) = b^{2}(t)\, G^{\mathrm{max}}.$$
Note that the identity component of  $\mathcal{E}_t$ does not contribute to the  averaged bipartite OTOC and $G(\mathcal{E}_\infty)=\frac{1}{4}G^{\mathrm{max}}.$

{\em{Example 2.--}}  Let us consider the  Lindbladian  $$\mathcal{L} = i\, \mathrm{ad} {H} + \lambda \left( \mathcal{D}_{\mathbb{B}} - I \right),$$ where \(\mathrm{ad} {H}(X) := \left[ H,X \right]\) and $\mathcal{D}_{\mathbb{B}}$  is the dephasing superoperator.
We assume that the dephasing basis is the same as the Hamiltonian eigenbasis, i.e., \(\mathbb{B} = \{ \Pi_{j} \}\) with \(\Pi_{j}\) the Hamiltonian eigenstates. 
In this case $[\mathrm{ad} H, \mathcal{D}_{\mathbb{B}}]=0$ and 
therefore  the dynamics  is given by a convex combination of a unitary and a dephasing channel
  $$\mathcal{E}_t  = e^{t \mathcal{L}}=\tilde{a}(t)\,  e^{it \mathrm{ad} {H}} + \tilde{b}(t)\,  \mathcal{D}_{\mathbb{B}}$$
where $ \tilde{a}(t):=e^{- \lambda t}$ and $\tilde{b}(t):=1 -\tilde{a}(t).$
This corresponds to  $ \hat{\Phi}_{\alpha\alpha'}=  \tilde{a}(t)\,e^{it(\lambda_\alpha-\lambda_{\alpha'})} e +    \tilde{b}(t)\, \delta_  {\alpha\alpha'},$   
with  $\lambda_{\alpha/\alpha'}=\pm 1.$
Moreover, if we assume that \(\mathcal{D}_{\mathbb{B}} \left( X \otimes \frac{\mathbb{I}}{d_{B}} \right) = \operatorname{Tr}\left( X \right) \frac{\mathbb{I}}{d}\). Then, the bipartite OTOC becomes,
  $$G(\mathcal{E}_t) = \tilde{a}^2(t)\, G(e^{it \mathrm{ad} {H}} ).$$
 If  the Hamiltonian is the swap operator, $S$
one gets $G(e^{it \mathrm{ad} {H}} )= \left( 1 - \cos^{4}(t)\right)\, G^{\mathrm{max}}.$ Fig.~\ref{fig:swap-w-dissi} shows the corresponding   pattern of exponentially damped oscillations.

\section{Quantum Spin Chains}
\label{sec:spin-chains}

\begin{figure*}[!t]
\raggedright
\begin{subfigure}{.44\textwidth}
  \includegraphics[width=1.3\linewidth]{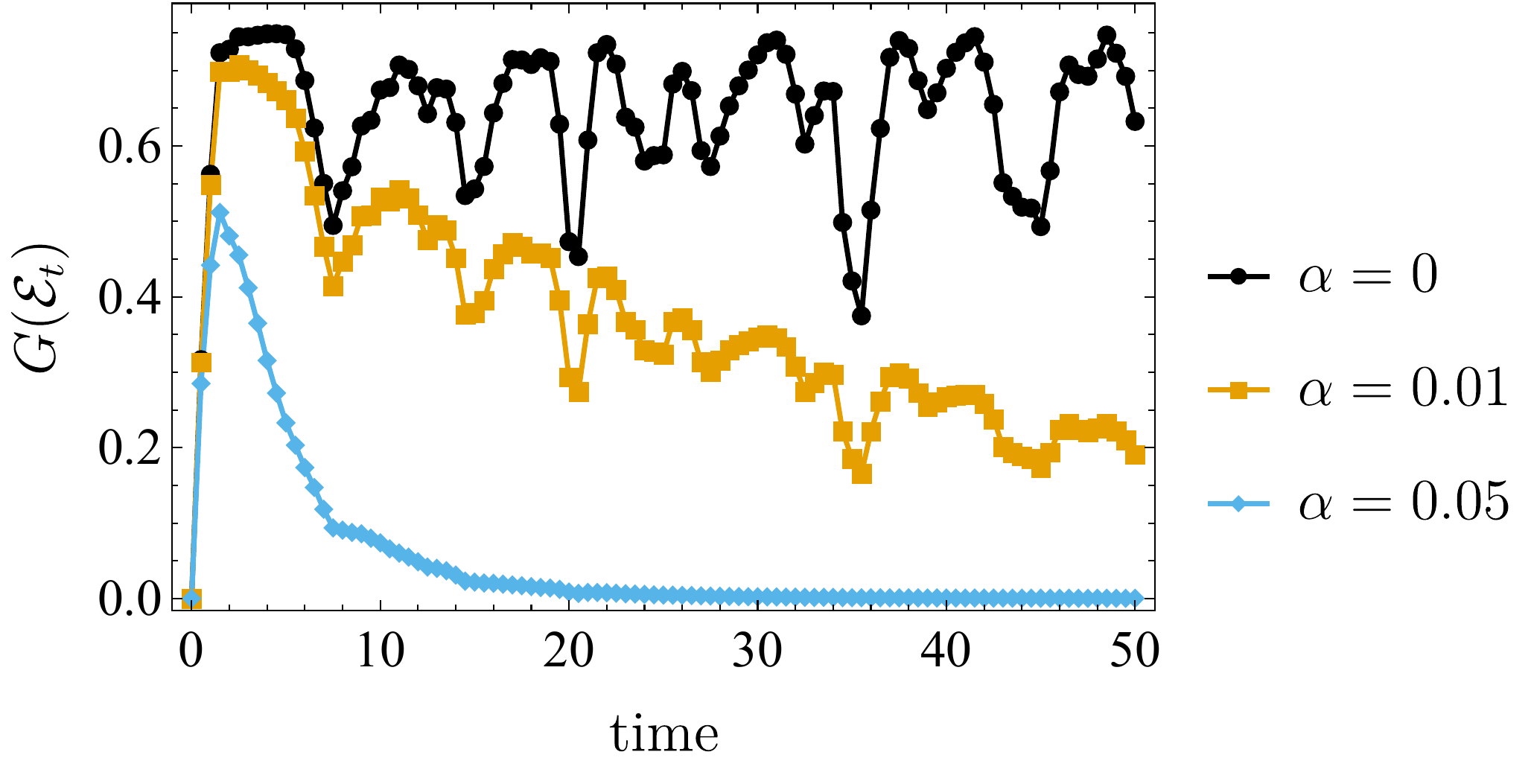}
  \caption{integrable}
\end{subfigure}\hspace{65pt}
\begin{subfigure}{.34\textwidth}
  \includegraphics[width=1.3\linewidth]{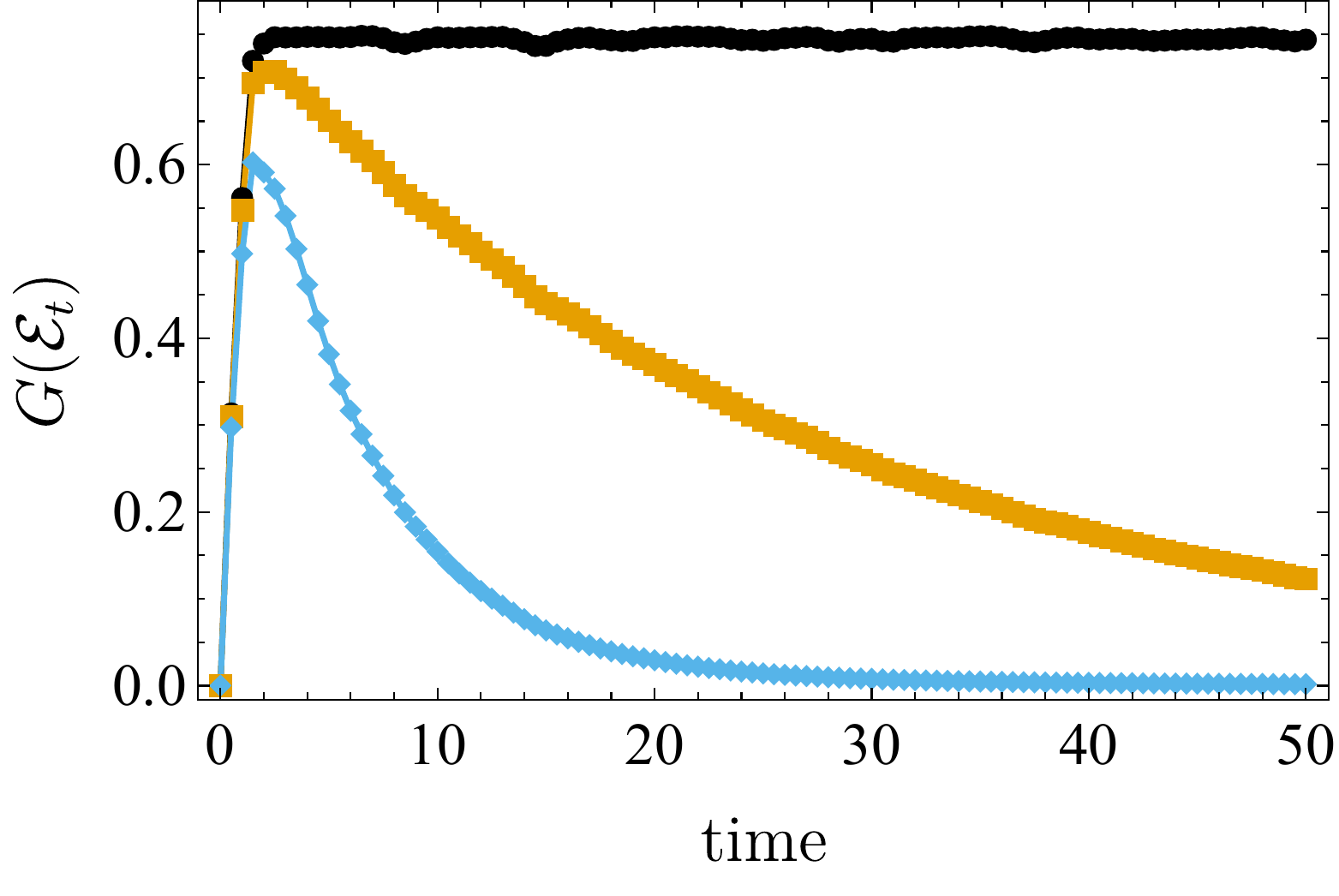}
    \caption{chaotic}
\end{subfigure}%
\caption{Temporal variation of the open OTOC $G(\mathcal{E}_{t})$ for the TFIM~\cref{eq:tfim} with $L=6$ spins. The three curves correspond to varying choices of the dissipation strength $\alpha$ in the Lindblad operators~\cref{eq:lindblad-gksl}. The chaotic ($g=-1.05,h=0.5$) and integrable ($g=1,h=0$) phases are clearly distinguishable for the $\alpha=0$ case (closed system), however, increasing the dissipation strength to $\alpha=0.05$ makes them fairly indiscernible and destroys the revivals (or fluctuations) characteristic of integrable systems.}
\label{fig:TFIM-6qubits-openotoc}
\end{figure*}

\begin{figure*}[!t]
\raggedright
\begin{subfigure}{.45\textwidth}
  \includegraphics[width=1.3\linewidth]{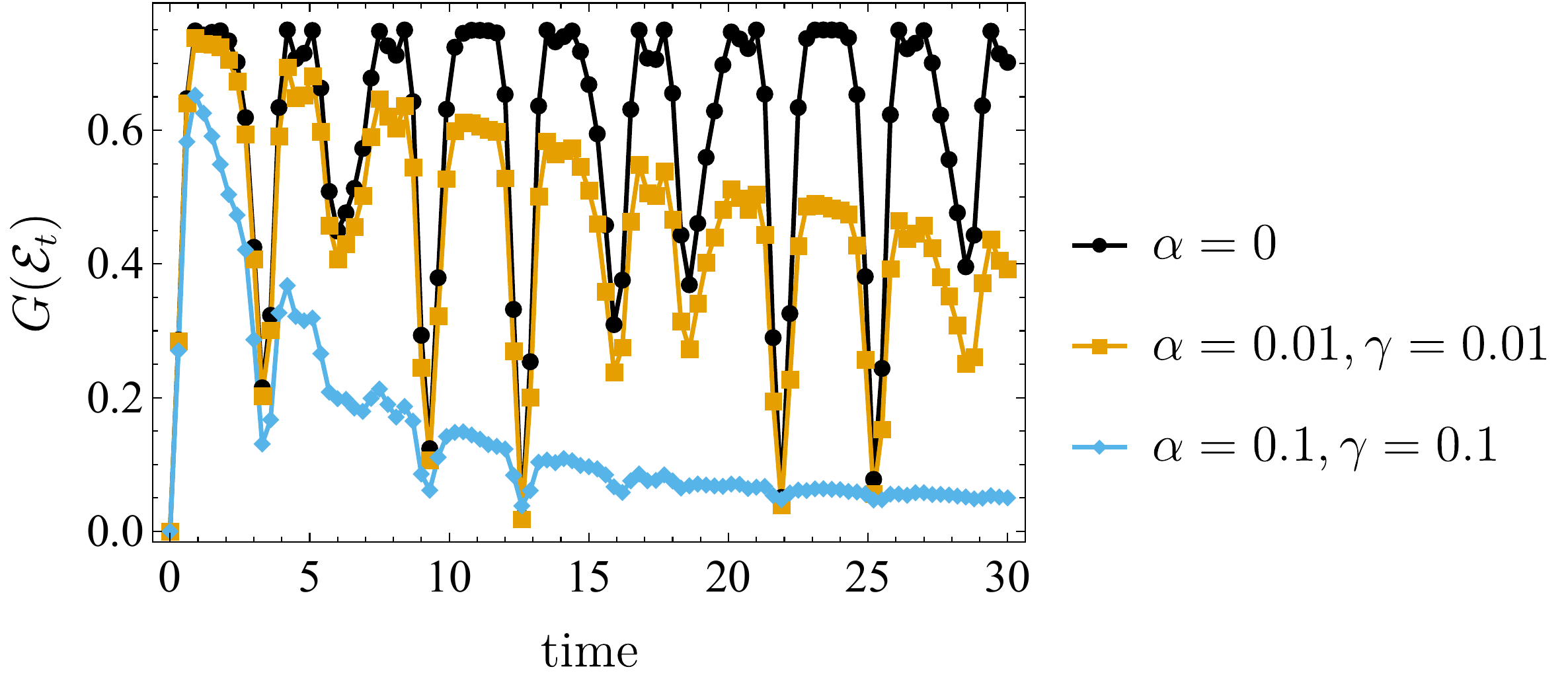}
  \caption{integrable}
\end{subfigure}\hspace{65pt}
\begin{subfigure}{.31\textwidth}
  \includegraphics[width=1.3\linewidth]{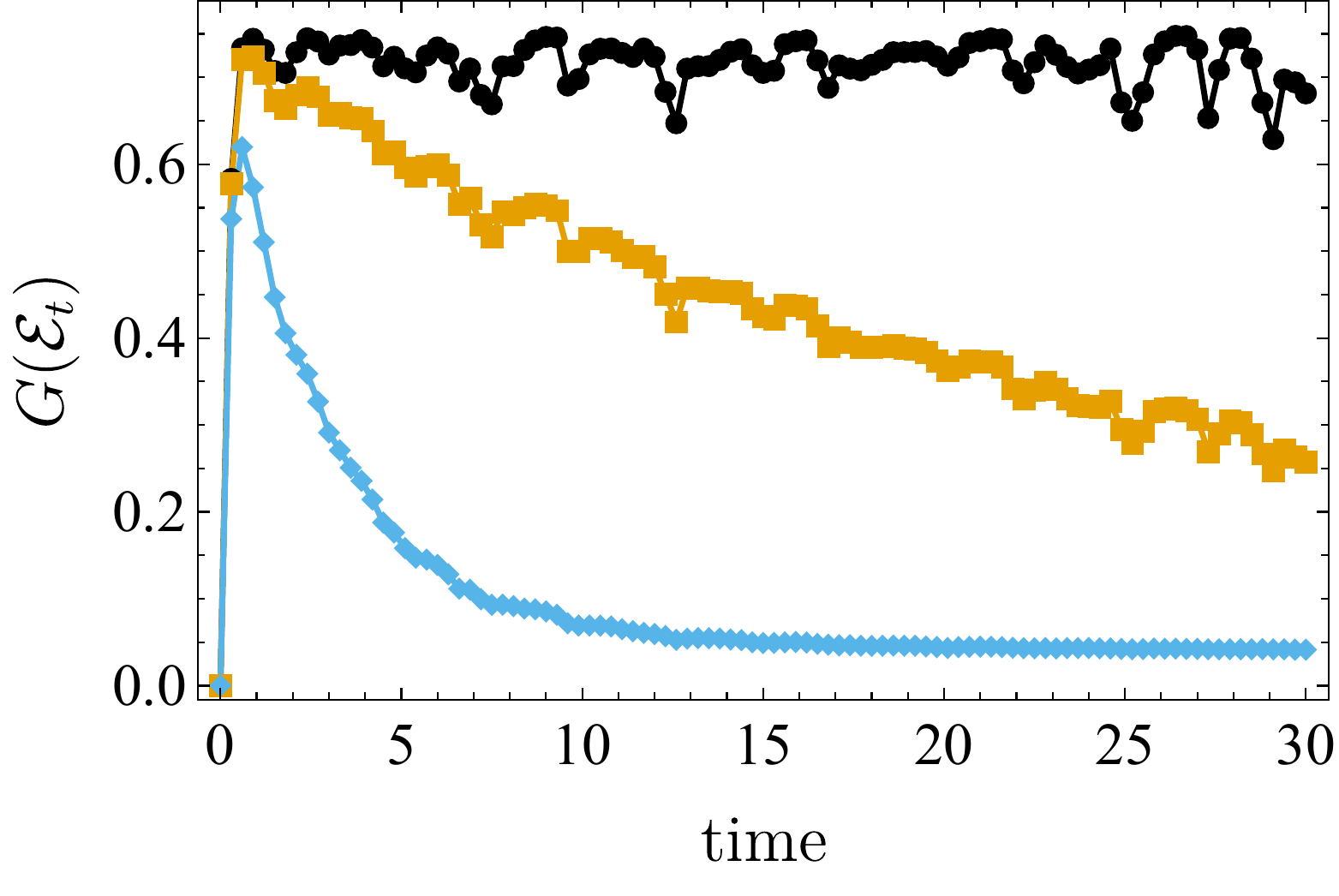}
    \caption{chaotic}
\end{subfigure}%
\caption{Temporal variation of the open OTOC $G(\mathcal{E}_{t})$ for the XXZ-NNN model~\cref{eq:xxz-nnn} with $L=6$ spins. The three curves correspond to varying choices of the dissipation strength $\alpha$ in the Lindblad operators~\cref{eq:lindblad-gksl}
. The nonintegrable ($J=1,\Delta=0.5,J=1,\Delta'=0.5$) and integrable ($J=1,\Delta=0=J=\Delta'$) phases are clearly distinguishable for the $\alpha=0=\gamma$ case (closed system). The integrable model here can be mapped onto free fermions and hence unlike the TFIM case, even after increasing the dissipation strength ($\alpha=0.1=\gamma$), the system demonstrates revivals (or fluctuations) characteristic of integrable systems.}
\label{fig:XXZ-6qubits-openotoc}
\end{figure*}

\begin{figure*}[!t]
\raggedright
\begin{subfigure}{.42\textwidth}
  \includegraphics[width=1.3\linewidth]{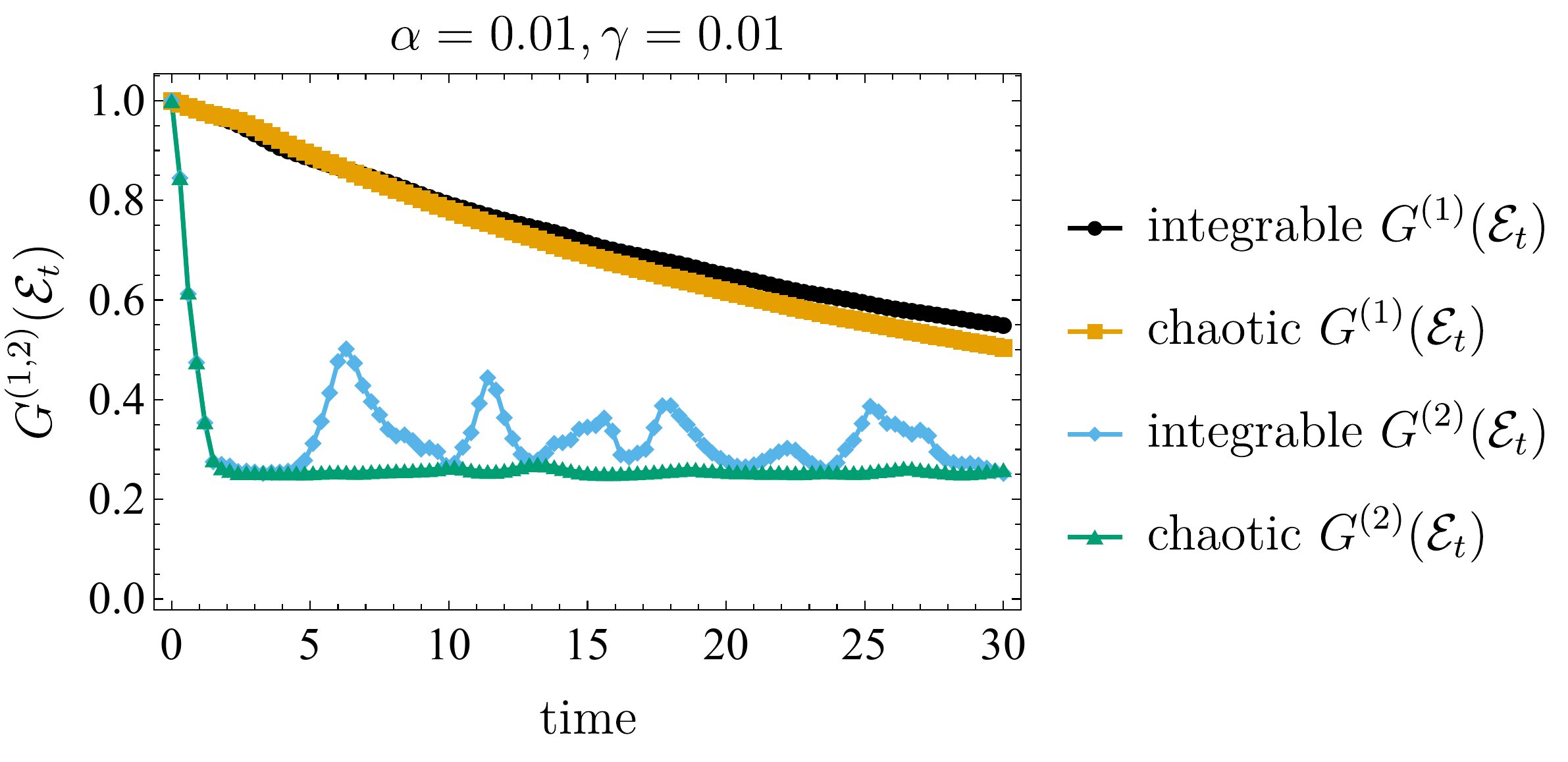}
  \caption{TFIM}
\end{subfigure}\hspace{65pt}
\begin{subfigure}{.29\textwidth}
  \includegraphics[width=1.3\linewidth]{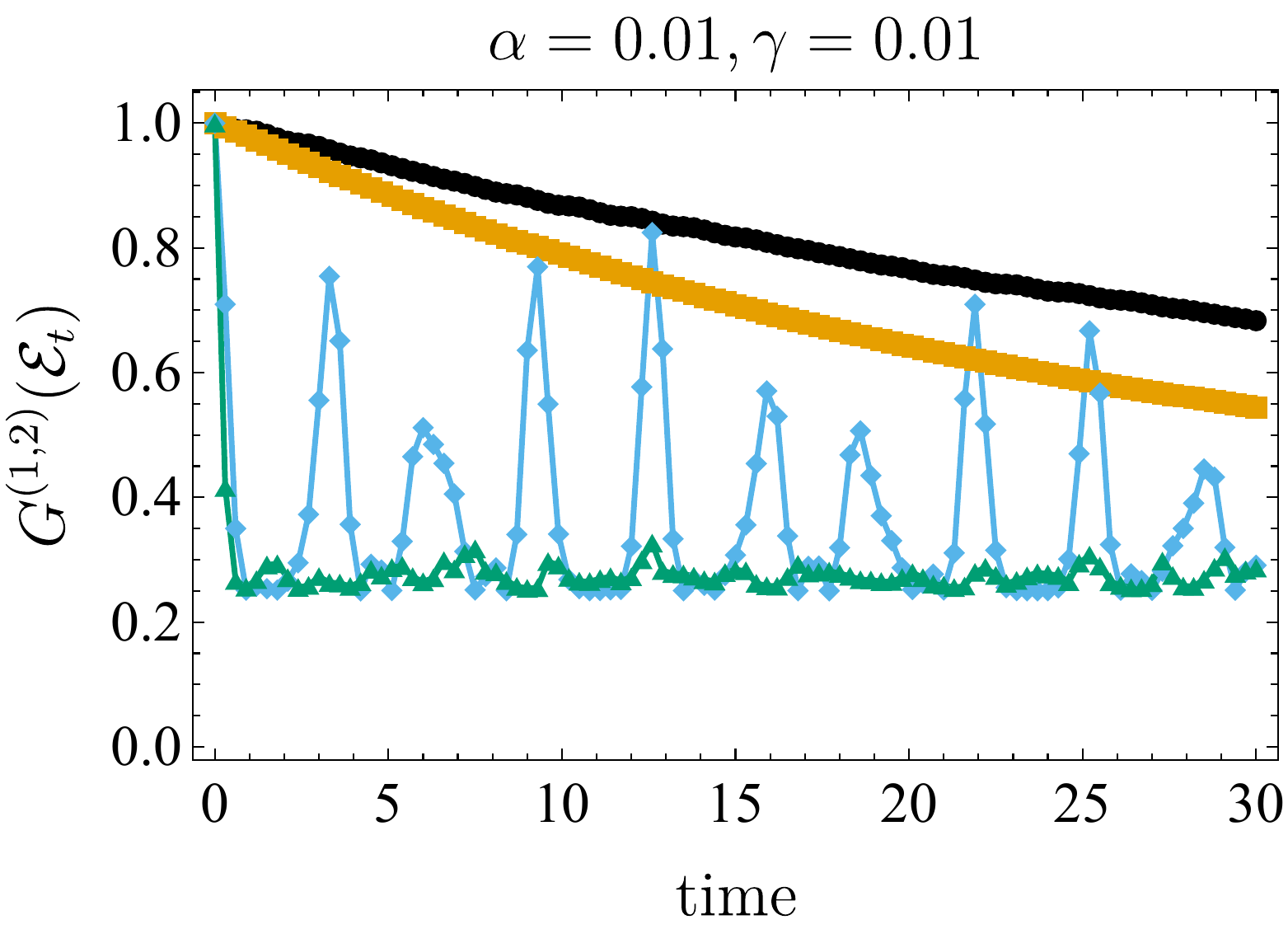}
    \caption{XXZ-NNN model}
\end{subfigure}%
\caption{Temporal variation of the individual terms of the open OTOC \(G^{(1)}(\mathcal{E}_{t})= \frac{d_{B}}{d^{2}} \operatorname{Tr}\left[ S \mathcal{E}^{\otimes 2} \left( S_{AA'} \right) \right]\) and \(G^{(2)}(\mathcal{E}_{t})= \frac{1}{d^{2}} \operatorname{Tr}\left[ S_{AA'} \mathcal{E}^{\otimes 2} (S_{AA'}) \right]\) with \(G(\mathcal{E}_{t}) = G^{(1)}(\mathcal{E}_{t}) -  G^{(2)}(\mathcal{E}_{t})\). The two figures correspond to the integrable and chaotic limits as considered above for the (a) TFIM and (b) XXZ-NNN model with $L=6$ spins, respectively. The dissipation parameters are \(\alpha=0.01,\gamma=0.01\). The first term \(G^{(1)}(\mathcal{E}_{t})\) originates from environmental decoherence and is similar for both the integrable and the chaotic case. However, the second term, \(G^{(2)}(\mathcal{E}_{t})\) is clearly distinct for the two phases and can diagnose quantum chaos even in the presence of dissipation.}
\label{fig:first-second-terms-otoc}
\end{figure*}

As a physical application of the open OTOC, we study paradigmatic quantum spin-chain models of quantum chaos in the presence of open-system dynamics. For systems interacting with a Markovian environment, the dynamics can be described by a Lindblad master equation (sometimes also called the GKSL form)~\cite{breuerTheoryOpenQuantum2002},
\begin{align}
\label{eq:lindblad-gksl}
&\frac{d \rho(t)}{d t} = \mathcal{L}^\dagger (\rho(t)) \equiv -i \left[ H, \rho(t) \right] \nonumber \\
&+ \sum\limits_{j}^{} \left( L_{j} \rho(t) L^{\dagger}_{j} - \frac{1}{2} \{L_{j}^{\dagger}L_{j}, \rho(t)\} \right),
\end{align}
where $\mathcal{L}^\dagger$ is the Lindbladian, \(H\) is the Hamiltonian, \(\rho(t)\) is the quantum state at time \(t\), and \(\{ L_{j} \}\) are called the Lindblad (or jump) operators, which constitute the system-environment interaction. The master equation above gives rise to a one-parameter family of time-evolution superoperators (in the Schr\"odinger picture),
\begin{align}
\mathcal{E}^{\dagger}_{t} = e^{t \mathcal{L}^{\dagger}}, t \geq 0.
\end{align}

We consider two quantum spin-\(1/2\) chains on \(L\) sites, (i) the transverse-field Ising model (TFIM) with an onsite magnetization and (ii) the next-to-nearest neighbor Heisenberg XXZ model (XXZ-NNN).
\begin{align}
\label{eq:tfim}
H_{\mathrm{TFIM}} &= - \left( \sum\limits_{j}^{} \sigma_{j}^{z} \sigma_{j+1}^{z} + g \sigma_{j}^{x} + h \sigma_{j}^{z} \right).
\end{align}
\begin{align}
\label{eq:xxz-nnn}
H_{\mathrm{XXZ}} &= J \sum\limits_{j=1}^{L-1} \left( \sigma_{j}^{x} \sigma_{j+1}^{x} + \sigma_{j}^{y} \sigma_{j+1}^{y} + \Delta \sigma_{j}^{z} \sigma_{j+1}^{z} \right) \nonumber \\ 
& + J' \sum\limits_{j=1}^{L-2} \left( \sigma_{j}^{x} \sigma_{j+2}^{x} + \sigma_{j}^{y} \sigma_{j+2}^{y} + \Delta' \sigma_{j}^{z} \sigma_{j+2}^{z} \right).
\end{align}
Here, the \(\sigma_{j}^{\alpha}, \alpha \in \{ x,y,z \}\) are the Pauli matrices. For the TFIM, \(g,h\) denotes the strength of the transverse field and the local field, respectively. The TFIM Hamiltonian is integrable for \(h=0\) and nonintegrable when both \(g,h\) are nonzero. We consider as the integrable point, \(g=1,h=0\) and the nonintegrable point \(g=-1.05,h=0.5\). For the XXZ-NNN model, \(J(J')\) denotes the strength of the nearest- (next-to-nearest-) neighbor coupling, and \(\Delta(\Delta')\) denotes the anisotropy along the \(z\)-axis. The XXZ-NNN model Hamiltonian is integrable by Bethe Ansatz for \(J'=0 = \Delta'\). We consider as the integrable point, $J=1,\Delta=0=J=\Delta'$ which can be mapped onto free fermions and as the nonintegrable point, $J=1,\Delta=0.5,J=1,\Delta'=0.5$~\cite{baxter2016exactly}.

We consider two types of jump processes at the boundary: (i) amplitude damping, with Lindblad operators \(\sqrt{\alpha} \sigma_{1}^{\pm}\) and \(\sqrt{\alpha} \sigma_{L}^{\pm}\); and (ii) boundary dephasing, with Lindblad operators \(\sqrt{\gamma} \sigma_{1}^{z}, \sqrt{\gamma} \sigma_{L}^{z}\). Note that similar models have been considered before to study non-equilibrium spin transport~\cite{Bu_a_2012,PhysRevLett.107.137201,PhysRevLett.117.137202} and dissipative quantum chaos~\cite{PhysRevX.10.021019}. To numerically simulate the evolution, we ``vectorize'' the Lindbladian superoperator \(\mathcal{L}\) into a \(4^{L} \times 4^{L}\) dimensional matrix representation,
\begin{align}
|\mathcal{L}\rangle\!\rangle =& i \left( \mathbb{I} \otimes H^{\dagger} - H^{*} \otimes \mathbb{I} \right) + \sum\limits_{j}^{} ( L_{j}^{T} \otimes L_{j}^{\dagger} \nonumber \\
& - \frac{1}{2} L_{j}^{*} L_{j}^{T} \otimes \mathbb{I} - \frac{1}{2} \mathbb{I} \otimes L_{j}^{\dagger}L_{j}),
\end{align}
where \(X^{T},X^{*}\) denotes the matrix transpose and complex conjugation, respectively \footnote{{This matrix representation is also sometimes known as the \textit{Liouville representation}. It is closely related to the Choi-Jamiolkowski form via, \(|\mathcal{E}\rangle\!\rangle^{R} = \rho_{\mathcal{E}}\), where \(|jk\rangle\langle lm|^{R} = |jl\rangle\langle km|\) for all basis states \(\{|j\rangle\}\) is known as the ``reshuffling'' operation.}}.

We simulate exact dynamics for the open system and compute \(G(\mathcal{E}_{t})\) for \(L=6\) spins across the bipartition \(1:L-1\). In~\cref{fig:TFIM-6qubits-openotoc,fig:XXZ-6qubits-openotoc} we consider the two models in~\cref{eq:tfim,eq:xxz-nnn} with their integrable and chaotic limits. As we increase the strength of system-environment coupling, namely, in~\cref{fig:TFIM-6qubits-openotoc} the parameter $\alpha$ and in~\cref{fig:XXZ-6qubits-openotoc} the parameters $\alpha,\gamma$, the open OTOC $G(\mathcal{E}_{t})$ starts decaying from its closed system value, $G(\mathcal{U}_{t})$. In~\cref{fig:TFIM-6qubits-openotoc} the integrable and chaotic phases are clearly distinguishable for the closed system case ($\alpha=0$), however, for $\alpha=0.05$, the phases become indiscernible due to open-system effects. Similarly, in~\cref{fig:XXZ-6qubits-openotoc}, the revivals in the free fermions regime is clearly distinguishable from the nonintegrable regime for the closed system ($\alpha=0=\gamma$). However, at $\alpha=0.1=\gamma$, the two are less discernible. Note, however, in this ``strongly integrable'' regime (since the system can be mapped onto free fermions), even by increasing the dissipation strength, one can see revivals (or fluctuations).

Furthermore, following the intuition developed in \autoref{th:botoc-op-ent-open} we can separate the contributions due to environmental decoherence and the dynamical entanglement generation. The open OTOC  \(G(\mathcal{E}_{t}) = G^{(1)}(\mathcal{E}_{t}) -  G^{(2)}(\mathcal{E}_{t})\) is the difference of two terms, \(G^{(1)}(\mathcal{E}_{t}) \equiv \frac{d_{B}}{d^{2}} \operatorname{Tr}\left[ S \mathcal{E}^{\otimes 2} \left( S_{AA'} \right) \right]\) and \(G^{(2)}(\mathcal{E}_{t}) \equiv \frac{1}{d^{2}} \operatorname{Tr}\left[ S_{AA'} \mathcal{E}^{\otimes 2} (S_{AA'}) \right]\). As illustrated in~\cref{fig:first-second-terms-otoc}, the first term $G^{(1)}(\mathcal{E}_{t})$ displays a similar behavior in the integrable and chaotic regimes for both the TFIM and the XXZ-NNN model, however, the second term, $G^{(2)}(\mathcal{E}_{t})$ can still diagnose quantum chaos, even in the presence of dissipation. In fact, as we know from \autoref{th:Choi-open}, this is the open-system variant of the operator entanglement-OTOC connection for the unitary case and is expected to be the diagnostic of these two phases. Moreover, notice that after separating these two contributions, one is able to distinguish the chaotic and integrable phases for the TFIM which were less discernible previously.

\section{Conclusions}
\label{sec:conclusions}

In this work we generalize the bipartite OTOC to the case of open quantum dynamics described by quantum channels. We provide an exact analytical expression for this \textit{open} bipartite OTOC which allows us to understand the competing entropic contributions from environmental decoherence and information scrambling. The separate contributions to entropy production can be understood via (a) \autoref{th:Choi-open}, as the difference of purities of the Choi state (corresponding to the dynamical map) across different partitions and (b) in \autoref{th:entropy-production-open} as the average entropy production under the reduced dynamics and that due to the (global) mixedness of the evolution. 

As a concrete example, we study special classes of channels, namely dephasing channels, entanglement-breaking channels, and $\mathbb{B}$-diagonal channels. For dephasing channels, the open OTOC can be expressed in terms of the ``idempotency deficit'' of the Gram matrix of reduced states of the states in the dephasing basis. Moreover, if the (dephasing) basis states are highly entangled then an upper bound on the open OTOC can be obtained in terms of their deviation from maximal entanglement. Furthermore, we provide an analytical estimate of the open OTOC for random dephasing channels, deviations from which are exponentially suppressed due to measure concentration. 

Finally, as a physical application of our analytical results, we consider paradigmatic quantum spin-chain models of quantum chaos in the presence of open system dynamics. As expected, the dissipation effects obfuscate the dynamical scrambling of information and the integrable and chaotic phases become less discernible as the strength of dissipation is increased. However, our analytical results allow us to separate the entropic contributions making discernible the ``scrambling entropy'' even in the presence of dissipation.

In closing, we list two promising directions for future investigations. First, exploring further the interplay of the two distinct contributions to entropy production --- that can obfuscate the effect of information scrambling --- and how to build robust techniques to delineate them in an experimentally accessible way. And second, the averaged (bipartite) open OTOC discussed in this paper has a well defined quantum-information theoretic meaning in terms of operational protocols (see \cref{sec:experimental-openotoc}) which make no direct reference to the system's temperature. However, it is a  compelling topic for future research to study extensions where the formal infinite-temperature average is replaced by expectations over other steady states of the quantum channel under examination, for example, finite-temperature Gibbs states for Davies generators \cite{anand_zanardi_finitetemp} \footnote{A finite-temperature generalization of Proposition 1 for the case of unitary evolutions was already reported in the Supplemental Material of  Ref.~\cite{styliaris_information_2020}, following the proof of Theorem 1.}.


\section{Acknowledgments}
P.Z. acknowledges partial support from the NSF award PHY-1819189. This research was (partially) sponsored by the Army Research Office and was accomplished under Grant Number W911NF-20-1-0075. The views and conclusions contained in this document are those of the authors and should not be interpreted as representing the official policies, either expressed or implied, of the Army Research Office or the U.S. Government. The U.S. Government is authorized to reproduce and distribute reprints for Government purposes notwithstanding any copyright notation herein. 

\bibliography{refs,my_library}

\onecolumngrid
\widetext
\newpage

\appendix

\renewcommand{\thepage}{A\arabic{page}}
\setcounter{page}{1}
\renewcommand{\thesection}{A\arabic{section}}
\setcounter{section}{0}
\renewcommand{\thetable}{A\arabic{table}}
\setcounter{table}{0}
\renewcommand{\thefigure}{A\arabic{figure}}
\setcounter{figure}{0}
\renewcommand{\theequation}{A\arabic{equation}}
\setcounter{equation}{0}

\section{Review of operator entanglement and entangling power}
\label{sec:operator-entanglement-intro}

Let us briefly recall the ideas associated to operator entanglement and entangling power; see Refs.~\cite{zanardi2001entanglement,wang2002entanglement,zanardi2000entangling} for a detailed discussion. Given a \(d\)-dimensional Hilbert space \(\mathcal{H}\) the algebra of linear operators over \(\mathcal{H}\), \(\mathcal{L}(\mathcal{H})\) is endowed with a Hilbert space structure itself denoted as \(\mathcal{H}_{HS}\), induced via the Hilbert-Schmidt inner product. Moreover, \(\mathcal{H}_{HS}\) is isomorphic (both algebraically and as a Hilbert space) to \(\mathcal{H}^{\otimes 2}\), therefore, one can associate bipartite states to linear operators. This is analogous to the Choi-Jamiolkowski isomorphism.

Formally, given \(U \in \mathcal{H}_{HS}\), one can define \(| U \rangle := \left( U \otimes \mathbb{I} \right) | \Phi^{+} \rangle\), where \(| \Phi^{+} \rangle := \frac{1}{\sqrt{d}} \sum\limits_{j=1}^{d} | j \rangle | j \rangle\) is the maximally entangled state across \(\mathcal{H}^{\otimes 2}\). Now, if the Hilbert space \(\mathcal{H}\) itself has a bipartite structure, that is, \(\mathcal{H} \cong \mathcal{H}_{A} \otimes \mathcal{H}_{B}\) then the corresponding state-representation of $U \equiv U_{AB}$ (since it generically acts on the total space) is a four-party state. Namely, \(| U \rangle_{ABA'B'} = \left( U_{AB} \otimes \mathbb{I}_{A'B'} \right) | \Phi^{+} \rangle_{ABA'B'}\) with \(| \Phi^{+} \rangle_{ABA'B'} = \frac{1}{\sqrt{d}} \sum\limits_{j=1}^{d} | j \rangle_{AB} | j \rangle_{A'B'}\). Moreover, notice that for the state \(| U \rangle_{ABA'B'}\), the entanglement across the \(AB|A'B'\) partition is maximal (since it is local unitarily equivalent to the maximally entangled state). However, the entanglement across the \(AA'|BB'\) partition is nontrivial and one way to quantify this would be to compute the linear entropy across this bipartition. This is precisely the operator entanglement. That is, tracing out over \(BB'\) we obtain \(\sigma_{U} := \operatorname{Tr}_{BB'}\left[ | U \rangle \langle  U |  \right]\) and computing its linear entropy defined as \(S_{\mathrm{lin}}(\rho):= 1-\operatorname{Tr}\left[ \rho^{2} \right]\), we have,
\begin{align}
E_{\mathrm{op}}(U):= S_{\mathrm{lin}}(\sigma_{U}) = 1 - \operatorname{Tr}\left[ \left( \operatorname{Tr}_{BB'} | U \rangle \langle  U |  \right)^{2} \right].
\end{align}

Another key quantity that is related to the operator entanglement is the \textit{entangling power} of a unitary \(U\) acting on a (symmetric) bipartite space \(\mathcal{H}_{AB} \cong \mathcal{H}_{A} \otimes \mathcal{H}_{B}\) with \(d_{A} = d_{B} = \sqrt{d}\), defined as the average amount of entanglement generated by \(U\) via its action on pure product states. Formally,
\begin{align}
e_{p}(U) := \mathbb{E}_{V \in \mathcal{U}(H_{A}), W \in \mathcal{U}(\mathcal{H}_{B})} \left[ E_{\mathrm{op}} \left( U | \psi_{V_{A}} \rangle | \psi_{W_{B}} \rangle \right) \right],
\end{align}
where \(| \psi_{V_{A}} \rangle = V | \psi_{0} \rangle\) (and similarly for \(| \psi_{W_{B}} \rangle\)). Quite remarkably, the entangling power and the operator entanglement are related as,
\begin{align}
e_{p}(U) = \frac{d^{2}}{\left( d+1 \right)^{2}} \left[ E_{\mathrm{op}}(U) + E_{\mathrm{op}}(US) - E_{\mathrm{op}}(S) \right],
\end{align}
where \(S\) is the swap operator between subsystems \(A,B\) (assumed to be symmetric for the connection to entangling power).

\section{A protocol for estimating the Open OTOC}
\label{sec:experimental-openotoc}
\autoref{th:entropy-production-open} establishes the open OTOC, \(G(\mathcal{E})\) as the difference of two terms, each of which quantify the average entropy production of channels \(\widetilde{\mathcal{E}}\) and \(\operatorname{Tr}_{B}\left[ \widetilde{\mathcal{E}} \right]\), respectively. Let us briefly review the case for unitary channels first, which was first discussed in Ref. \cite{styliaris_information_2020}, see Section III of the Supplemental Material for more details.

For a unitary time evolution \(\{ \mathcal{U}_{t} \}_{t \geq 0}\), the bipartite OTOC can be expressed as,
\begin{align}
G(\mathcal{U}_{t}) = \frac{d_{A}+1}{d_{A}} \mathbb{E}_{\psi \in \mathcal{H}_{A}} \left[ S_{\mathrm{lin}} \left( \Lambda_{t}^{(A)}(\psi) \right) \right],
\end{align}
where $\Lambda_{t}^{(A)}\left(\rho_{A}\right):=\operatorname{Tr}_{B}\left[U_{t}\left(\rho_{A} \otimes I_{B} / d_{B}\right) U_{t}^{\dagger}\right]$, \(\mathbb{E}_{\psi \in \mathcal{H}_{A}}\) denotes random pure states uniformly distributed in \(\mathcal{H}_{A}\), and \(S_{\mathrm{lin}}(\cdot)\) is the linear entropy. The basic protocol is to (i) initialize a random state in subsystem \(A\) and a maximally mixed state in subsystem \(B\), (ii) apply the channel \(\mathcal{U}_{t}\) to the entire system \(AB\), (iii) trace out subsystem \(B\), (iv) measure the linear entropy of the resulting state, and (v) repeat for many random initial states uniformly distributed in \(\mathcal{H}_{A}\).

The key idea is that (i) due to measure concentration, as \(d_{A}\) grows, fewer random states are needed to estimate \(G(\mathcal{U}_{t})\) exponentially well, and (ii) linear entropy of a quantum state can be measured in an experimentally accessible way, see, for example, the seminal experiment in Ref.~\cite{islam2015measuring} where the purity (which is equal to one minus the linear entropy) was measured by interfering two uncorrelated but identical copies of a many-body quantum state; similar ideas have also been considered previously~\cite{daley2012measuring,ekert2002direct,PhysRevLett.93.110501,bovino2005direct}.

Furthermore, there have also been recent proposals based on measurements over random local bases that can probe entanglement given just a single copy of the quantum state, and, in this sense, go beyond traditional quantum state tomography. The main idea consists of directly expressing the linear entropy~\cite{brydges2019probing,elben2019statistical}, as well as other functions of the state~\cite{huang2020predicting}, as an ensemble average of measurements over random bases.\\

Now, for the open-system case, to estimate \(\mathbb{E}_{\psi} \left[ S_{\mathrm{lin}} \left( \operatorname{Tr}_{B}\widetilde{\mathcal{E}} \left( \psi \right) \right) \right]\), we replace in the protocol above, \(\mathcal{U}_{t}\) with the channel \(\mathcal{E}\). To understand this, recall that \(\widetilde{\mathcal{E}}\) is defined such that its action is \(\widetilde{\mathcal{E}}(\rho_{A}) \mapsto \mathcal{E} \left( \rho_{A} \otimes \frac{\mathbb{I}}{d_{B}} \right)\), analogous to the channel \(\Lambda_{t}^{(A)}\) for the unitary case. For the second term that is proportional to \(\mathbb{E}_{\psi} \left[ S_{\mathrm{lin}} \left(\widetilde{\mathcal{E}} \left( \psi \right) \right) \right]\), we simply drop the partial tracing over subsystem \(B\) above and everything else in the protocol is the same.\\

\section{Proofs}
\label{sec:appendix}

\subsection*{Proof of \autoref{th:botoc-op-ent-open}}
\haaravgopenotoc*
\begin{proof}

Let us start by simplifying \(G(\mathcal{E})\) first,
\begin{align}
  G(\mathcal{E}) := \frac{1}{2 \left( d_{A} d_{B} \right)} \mathbb{E}_{A \in \mathcal{U}_A, B \in \mathcal{U}_B} \left\Vert \left[ \mathcal{E}(A), B \right]  \right\Vert_{2}^{2}, 
\end{align}
where \(\mathbb{E}_{X \in \mathcal{G}} \left( \cdots \right) \equiv \int_{X\sim \mathrm{Haar}} dX\left( \cdots \right)\) denotes the Haar average and the factor of \(\frac{1}{2}\) originates from the squared commutator, while the factor of \(\frac{1}{d_{A} d_{B}}\) is for the infinite-temperature state. 

Now, expanding the commutator gives us,
\begin{align}
  G(\mathcal{E}) = \frac{1}{d} \left[ \operatorname{Tr}\left( \mathcal{E}(A^{\dagger}) \mathcal{E}(A) \right) - \mathrm{Re} \operatorname{Tr}\left( \mathcal{E}(A) B^{\dagger} \mathcal{E}(A^{\dagger}) B \right) \right] ,
\end{align}
where \(d \equiv d_{A} d_{B}\).

Using the identity,
\begin{align}
  \operatorname{Tr}\left( XY \right) = \operatorname{Tr}\left( S X \otimes Y \right),
\end{align}
we have,
\begin{align}
  \operatorname{Tr}\left( \mathcal{E}(A^{\dagger}) \mathcal{E}(A) \right) = \operatorname{Tr}\left( S \mathcal{E}^{\otimes 2} (A^{\dagger} \otimes A) \right).
\end{align}
And,
\begin{align}
\operatorname{Tr}\left( \mathcal{E}(A) B^{\dagger} \mathcal{E}(A^{\dagger}) B  \right) = \operatorname{Tr}\left( S \mathcal{E}^{\otimes 2} \left( \left( A \otimes A^{\dagger} \right) \right) \left( B^{\dagger} \otimes B \right) \right)
\end{align}

We now use another key identity,
\begin{align}
  \mathbb{E}_{A \in \mathcal{U}_{A}} \left( A^{\dagger} \otimes A \right) = \frac{S_{AA'}}{d_{A}},
\end{align}
where \(S_{AA'}\) is the operator that swaps the replicas \(A\) with \(A'\). The analogous expression for \(BB'\) also holds.

Then, we have,
\begin{align}
  \mathbb{E}_{A \in \mathcal{U}_{A}} \operatorname{Tr}\left(S \mathcal{E}^{\otimes 2} \left( A^{\dagger} \otimes A \right) \right) = \operatorname{Tr}\left( S \mathcal{E}^{\otimes 2} \left( \frac{S_{AA'}}{d_{A}} \right) \right),
\end{align}

and,
\begin{align}
\mathbb{E}_{A \in \mathcal{U}_{A}, B \in \mathcal{U}_{B}} \operatorname{Tr}\left( \mathcal{E}(A) B^{\dagger} \mathcal{E}(A^{\dagger}) B  \right) = \operatorname{Tr}\left( S \mathcal{E}^{\otimes 2} \left( \frac{S_{AA'}}{d_{A}} \right) \left( \frac{S_{BB'}}{d_{B}} \right) \right) = \frac{1}{d} \operatorname{Tr}\left( S_{AA'} \mathcal{E}^{\otimes 2} \left( S_{AA'} \right) \right),
\end{align}
where in the last equality we have used the fact that \(S = S_{AA'} S_{BB'}\).

Putting everything together, we have,
\begin{align}
 G(\mathcal{E}) = \frac{1}{\left( d_{A} d_{B} \right)^{2}} \operatorname{Tr}\left( \left( d_{B} S - S_{AA'} \right) \mathcal{E}^{\otimes 2} (S_{AA'}) \right).
\end{align}

Note that if \(\mathcal{E} = \mathcal{U}\), then, \(S U^{\otimes 2} S_{AA'} U^{\dagger \otimes 2} = U^{\otimes 2} S_{BB'} U^{\dagger \otimes 2}\) using \(\left[ S, U^{\otimes 2} \right] =0\) and \(S = S_{AA'} S_{BB'}\), then, the first term of \(G(\mathcal{E})\) becomes one.

\begin{center}
\noindent\rule[0.5ex]{0.5\linewidth}{0.8pt}
\end{center}

We now show that \(\left[ L_{S}, \mathcal{E}^{\otimes 2} \right] = 0 \iff \mathcal{E} \text{ is unitary.} \)

Let \(L_{S}(X):= SX\) be the superoperator that denotes the left action of the swap operator. Note that,
\begin{align}
  L_{S} \mathcal{E}^{\otimes 2}(S) = \mathcal{E}^{\otimes 2} L_{S}(S) = \mathcal{E}^{\otimes 2}(I).
\end{align}

And, let \(\mathcal{E}(X) = \sum\limits_{j}^{} A_{j} X A_{j}^{\dagger}\) with \(\sum\limits_{j}^{} A_{j}^{\dagger} A_{j} = I\) be its Kraus representation. Then,
\begin{align}
  \text{LHS = } S \left( \sum\limits_{i,j}^{} \left( A_{i} \otimes A_{j} \right) S \left( A^{\dagger}_{i} \otimes A^{\dagger}_{j} \right) \right) = \sum\limits_{i,j}^{} A_{j} A^{\dagger}_{i} \otimes A_{i} A^{\dagger}_{j}.
\end{align}
Now, the RHS is \(\sum\limits_{i,j}^{} A_{i} A^{\dagger}_{i} \otimes A_{j} A^{\dagger}_{j}\). Taking the trace of both sides, we have,
\begin{align}
  \sum\limits_{i,j}^{} \left| \operatorname{Tr}\left( A_{j} A^{\dagger}_{i} \right) \right|^{2} = \sum\limits_{i,j}^{} \underbrace{\operatorname{Tr}\left( A_{i} A^{\dagger}_{i} \right)}_{\left\Vert A_{i} \right\Vert_{2}^{2}} \underbrace{\operatorname{Tr}\left( A_{j} A^{\dagger}_{j}\right)}_{\left\Vert A_{j} \right\Vert_{2}^{2}}.
\end{align}

Using Cauchy-Schwarz inequality, we have,
\begin{align}
  \left| \left\langle A_{j}, A_{i} \right\rangle \right|^{2} \leq \left\Vert A_{i} \right\Vert_{2}^{2} \left\Vert A_{j} \right\Vert_{2}^{2},  
\end{align}
where the equality holds if and only if \(~~\forall i,j A_{i} =  \lambda_{ij} A_{j}\).

Say \(A_{i} = \lambda_{i} A_{0} ~~\forall i\). Then,
\begin{align}
  \mathcal{E}(X) = \sum\limits_{i}^{} \left| \lambda_{i} \right|^{2} A_{0} X A^{\dagger}_{0} = \widetilde{A}_{0} X \widetilde{A}^{\dagger}_{0}, \text{ where } \widetilde{A}_{0} \equiv \sqrt{\sum\limits_{i}^{} \left| \lambda_{i} \right|^{2}} A_{0}.
\end{align}

Namely, \(\mathcal{E}\) is a CP map with a single Kraus operator, therefore, \(\mathcal{E}\) is unitary.

\end{proof}

\subsection*{Proof of \autoref{th:Choi-open}}
\choiopen*
\begin{proof}
Given a bipartite channel, \(\mathcal{E}_{AB}\), consider its Choi state,
\begin{align}
\rho_{\mathcal{E}} = \left( \mathcal{E}_{AB} \otimes I_{A'B'} \right) \left( | \phi^{+} \rangle \langle  \phi^{+} |  \right),
\end{align}
where \(| \phi^{+} \rangle \equiv | \phi^{+}_{ABA'B'} \rangle = \frac{1}{\sqrt{d}} \sum\limits_{i,j=1}^{d_{A} d_{B}} | i_{A} j_{B} \rangle \otimes | i_{A'} j_{B'} \rangle\).

Then,
\begin{align}
\rho_{\mathcal{E}} = \frac{1}{d} \sum\limits_{i,l=1}^{d_{A}} \sum\limits_{j,m=1}^{d_{B}} \mathcal{E} \left( | i \rangle \langle  l | \otimes | j \rangle \langle  m |  \right) \otimes | i \rangle \langle  l | \otimes | j \rangle \langle  m | .
\end{align}

Notice,
\begin{align}
\operatorname{Tr}_{BB'}\left[\rho_{\mathcal{E}}\right] = \frac{1}{d} \sum\limits_{i,l=1}^{d_{A}} \operatorname{Tr}_{B}\left[ \mathcal{E} \left( | i \rangle \langle  l | \otimes I_{B} \right) \otimes | i \rangle \langle  l |  \right] \equiv \rho_{\mathcal{E}}^{AA'}.
\end{align}
And,
\begin{align}
\operatorname{Tr}_{B'}\left[ \rho_{\mathcal{E}} \right] = \frac{1}{d} \sum\limits_{i,l=1}^{d_{A}} \mathcal{E} \left( | i \rangle \langle  l | \otimes I_{B} \right) \otimes | i \rangle \langle  l | \equiv \rho_{\mathcal{E}}^{ABA'}.
\end{align}
Then,
\begin{align}
\left\Vert \rho_{\mathcal{E}}^{AA'} \right\Vert_{2}^{2} = \frac{1}{d^{2}} \sum\limits_{i,l=1}^{d_{A}} \left\Vert \operatorname{Tr}_{B}\left[ \mathcal{E} \left( | i \rangle \langle  l | \otimes I_{B} \right) \right] \right\Vert_{2}^{2} = \frac{1}{d^{2}} \operatorname{Tr}\left[ S_{AA'} \mathcal{E}^{\otimes 2} (S_{AA'}) \right].
\end{align}
And,
\begin{align}
\left\Vert \rho_{\mathcal{E}}^{ABA'} \right\Vert_{2}^{2} = \frac{1}{d^{2}} \sum\limits_{i,l=1}^{d_{A}} \left\Vert \mathcal{E} \left( | i \rangle \langle  l | \otimes I_{B} \right) \right\Vert_{2}^{2} = \frac{1}{d^{2}} \operatorname{Tr}\left[ S \mathcal{E}^{\otimes 2} \left( S_{AA'} \right) \right].
\end{align}

Therefore, the open OTOC can be reexpressed as the difference of purities of the Choi state \(\rho_{\mathcal{E}}\) across different partitions,
\begin{align}
G(\mathcal{E}) = d_{B} \left\Vert \operatorname{Tr}_{B'}\left[ \rho_{\mathcal{E}} \right] \right\Vert_{2}^{2} - \left\Vert \operatorname{Tr}_{BB'}\left[ \rho_{\mathcal{E}} \right] \right\Vert_{2}^{2} .
\end{align}

Notice that for a dephasing channel, \(\mathcal{D}_{\mathbb{B}}\), one finds,
\begin{align}
\rho_{\mathcal{D}_{\mathbb{B}}} = \frac{1}{d} \sum\limits_{\alpha}^{} \Pi_{\alpha} \otimes \Pi_{\alpha} \equiv R_{\mathbb{B}}
\end{align}
and the \(G(\mathcal{D}_{\mathbb{B}})\) becomes the known expression for with the ``R-matrix''.

Moreover, for unitary channels, \(\rho_{\mathcal{U}}^{ABA'}\) is isospectral to \(\rho_{\mathcal{U}}^{B'}\) since the state \(\rho_{\mathcal{U}}^{ABA'B'}\) is pure. And, it is easy to show that \(\rho_{\mathcal{U}}^{B'} = I_{B'}/d_{B'}\), therefore, its purity is \(1/d_{B}\). That is, for unitary channels the first term of \(G(\mathcal{U})\) is equal to one, as expected. As a result, we have, \(G(\mathcal{U}) = 1 - \left\Vert \rho_{\mathcal{U}}^{AA'} \right\Vert_{2}^{2}\), which is the operator entanglement of the unitary channel \(\mathcal{U}\).

\begin{center}
\noindent\rule[0.5ex]{0.5\linewidth}{0.8pt}
\end{center}

To prove part (ii), notice that,
\begin{align}
&\operatorname{Tr}\left[ S \mathcal{E}^{\otimes 2} (S_{AA'}) \right] = \sum\limits_{i,j=1}^{d_{A}} \operatorname{Tr}\left[ S \mathcal{E}^{\otimes 2} \left( | i \rangle_{A} \langle  j | \otimes I_{B} \otimes | j \rangle_{A'} \langle  i | \otimes I_{B'}   \right) \right] = \sum\limits_{i,j=1}^{d_{A}} \operatorname{Tr}\left[ S \mathcal{E} \left( | i \rangle \langle  j | \otimes I_{B} \right) \otimes \mathcal{E} \left( | j \rangle \langle  i | \otimes I_{B'} \right) \right]\\
&= d_{B}^{2} \sum\limits_{i,j=1}^{d_{A}} \left\Vert \mathcal{E} \left( | i \rangle \langle  j | \otimes \frac{I_{B}}{d_{B}}  \right) \right\Vert_{2}^{2}.
\end{align}

Now, notice that, for \(\widetilde{\mathcal{E}}(X) = \mathcal{E} \left( X \otimes \frac{I_{B}}{d_{B}} \right)\), we have,
\begin{align}
\rho_{\widetilde{\mathcal{E}}} = \frac{1}{d_{A}} \sum\limits_{i,j}^{d_{A}} \left( \widetilde{\mathcal{E}} \otimes \mathcal{I}  \right) | i \rangle_{A} \langle  j | \otimes | i \rangle_{A'} \langle  j | = \frac{1}{d_{A}} \sum\limits_{i,j}^{d_{A}} \mathcal{E} \left( | i \rangle \langle  j | \otimes \frac{\mathbb{I}}{d_{B}} \right) \otimes | i \rangle_{A'} \langle  j |. 
\end{align}

Then, \(\left\Vert \rho_{\widetilde{\mathcal{E}}} \right\Vert_{2}^{2} = \frac{1}{d_{A}^{2}} \sum\limits_{i,j}^{d_{A}} \left\Vert \mathcal{E} \left( | i \rangle \langle  j | \otimes \frac{\mathbb{I}}{d_{B}} \right) \otimes | i \rangle_{A'} \langle  j | \right\Vert_{2}^{2}\).

Therefore,
\begin{align}
\frac{d_{B}}{d^{2}} \operatorname{Tr}\left[ S \mathcal{E}^{\otimes 2} \left( S_{AA'} \right) \right] = \frac{d_{B}}{d_{A}^{2}} \sum\limits_{ij}^{d_{A}} \left\Vert \mathcal{E} \left( | i \rangle \langle  j | \otimes \frac{I_{B}}{d_{B}} \right) \right\Vert_{2}^{2} = d_{B} \left\Vert \rho_{\widetilde{\mathcal{E}}} \right\Vert_{2}^{2}.
\end{align}

Similarly, we have,
\begin{align}
&\operatorname{Tr}\left[ S_{AA'} \mathcal{E}^{\otimes 2} \left( S_{AA'} \right) \right] = \frac{1}{d_{A}^{2}} \sum\limits_{i,j}^{d_{A}} \operatorname{Tr}\left[ S_{AA'} \mathcal{E} \left( | i \rangle \langle  j | \otimes \frac{I_{B}}{d_{B}} \right) \otimes \mathcal{E} \left( | j \rangle \langle  i | \otimes \frac{I_{B}}{d_{B}}  \right) \right] \\
&= \frac{1}{d_{A}^{2}} \sum\limits_{i,j}^{d_{A}} \left\Vert \operatorname{Tr}_{B}\left[ \mathcal{E} \left( | i \rangle \langle  j | \otimes \frac{I_{B}}{d_{B}} \right) \right] \right\Vert_{2}^{2} = d_{B} \left\Vert\rho_{\mathcal{T} \circ \widetilde{\mathcal{E}}} \right\Vert_{2}^{2} .
\end{align}

Putting everything together, we have the desired proof.
\end{proof}

\subsection*{Proof of \autoref{th:entropy-production-open}}
\entropyopen*
\begin{proof}

Let \(| \phi_{A} \rangle\) be an arbitrary state and \(| \psi_{A} \rangle := U | \phi_{A} \rangle\) correspond to Haar random pure states over \(\mathcal{H}_A\). Then, the key idea of the proof is the observation that \(S_{AA'}\) can be expressed via the identity,
\begin{align}
 \mathbb{E}_{\psi_{A} \sim \mathrm{Haar}} \left( | \psi_{A} \rangle \langle  \psi_{A} |  \right)^{\otimes 2} = \frac{1}{d_{A}(d_{A}+1)} \left( I_{AA'} + S_{AA'} \right).
\end{align}

Plugging this into Eq (2), we have,
\begin{align}
  1 - \frac{d_{A}+1}{d_{A}} \left\{ d_{B} \operatorname{Tr}\left( S \mathcal{E}^{\otimes 2} \left( \overline{\psi_{A}^{\otimes 2} \otimes \frac{I_{BB'}}{d_{B}^{2}}}^{\psi_{A}} \right)  \right) - \operatorname{Tr}\left( S_{AA'} \mathcal{E}^{\otimes 2}  \left( \overline{\psi_{A}^{\otimes 2} \otimes \frac{I_{BB'}}{d_{B}^{2}}}^{\psi_{A}} \right)  \right) \right\}.
\end{align}

Then, using,
\begin{align}
  S_{L}(X) = 1 - \operatorname{Tr}\left( X^{2} \right) = 1 - \operatorname{Tr}\left( S X \otimes X \right),
\end{align}
we have,
\begin{align}
  1- \frac{d_{A}+1}{d_{A}} \left\{ \overline{S_{L} \left( \operatorname{Tr}_{B}\left( \widetilde{\mathcal{E}} (\psi_{A}) \right)
 \right)}^{\psi_{A}} - d_{B} \left[ \overline{S_{L} \left( \widetilde{\mathcal{E}} (\psi_{A}) \right)}^{\psi_{A}} - \left( 1 - \frac{1}{d_{B}} \right) \right]  \right\},
\end{align}
where \(\widetilde{\mathcal{E}}(\psi_{A}) = \mathcal{E} \left( \psi_{A} \otimes \frac{I_{B}}{d_{B}} \right)\).

Now, notice that,
\begin{align}
  S_{L} \left( \widetilde{\mathcal{E}}(\psi_{A}) \right) \ge 1 - \frac{1}{d_{B}} \equiv S_{L}^{\mathrm{min}},
\end{align}
and since \(\mathcal{E}\) is unital, \(S_{L} \left( \mathcal{E} (\psi_{A} \otimes \frac{I_{B}}{d_{B}}) \right)\) must increase with time, since entropy cannot decrease under a unital map.
\end{proof}

\subsection*{Proof of \autoref{th:dephasing-channel}}
\dephasingchannel*
\begin{proof}
Consider the dephasing channel, \(\mathcal{E} = \mathcal{D}_{\mathbb{B}}\), where \(\mathcal{D}_{\mathbb{B}}(X) = \sum\limits_{\alpha=1}^{d} \Pi_{\alpha} X \Pi_{\alpha} = \sum\limits_{\alpha=1}^{d} | \psi_{\alpha} \rangle \langle  \psi_{\alpha} | \left\langle \psi_{\alpha} | X | \psi_{\alpha} \right\rangle\), where \(\{ \Pi_{\alpha} \}_{\alpha}\) is a basis (of rank-\(1\) projectors).

First, note that,
\begin{align}
  \operatorname{Tr}\left( S \mathcal{E}^{\otimes 2} (S_{AA'}) \right) = \sum\limits_{i,j=1}^{d_{A}} \operatorname{Tr}\left( S \mathcal{E} \left( | i \rangle \langle  j | \otimes I_{B} \right) \otimes \mathcal{E} \left( | j \rangle \langle  i | \otimes I_{B} \right) \right) = \sum\limits_{i,j=1}^{d_{A}} \left\Vert \mathcal{E} \left( | i \rangle \langle  j | \otimes I_{B} \right) \right\Vert_{2}^{2}, 
\end{align}
where we have used \(S_{AA'} = \sum\limits_{i,j=1}^{d_{A}} | ij \rangle_{AA'} \langle  ji | \otimes I_{BB'} \).

Now, 
\begin{align}
  \mathcal{E} \left( | i \rangle \langle  j | \otimes I_{B} \right) &= \sum\limits_{\alpha}^{} | \psi_{\alpha} \rangle \langle  \psi_{\alpha} | \left\langle \psi_{\alpha} | \left( | i \rangle \langle  j | \otimes I_{B} \right) | \psi_{\alpha} \right\rangle \\
&= \sum\limits_{\alpha}^{} \Pi_{\alpha} \operatorname{Tr}\left( \rho_{\alpha\alpha} | i \rangle \langle  j |  \right) = \sum\limits_{\alpha}^{} \left\langle j | \rho_{\alpha\alpha} | i \right\rangle \Pi_{\alpha},
\end{align}
where \(\rho_{\alpha \alpha} := \operatorname{Tr}_{B}\left( | \psi_{\alpha} \rangle \langle  \psi_{\alpha} |  \right)
\).

Therefore, 
\begin{align}
  \sum\limits_{i,j=1}^{d_{A}} \left\Vert \mathcal{E} \left( | i \rangle \langle  j | \otimes I_{B} \right) \right\Vert_{2}^{2} = \sum\limits_{i,j=1}^{d_{A}} \sum\limits_{\alpha}^{} \left| \left\langle i | \rho_{\alpha \alpha} | j \right\rangle \right|^{2} = \sum\limits_{\alpha}^{} \left\Vert \rho_{\alpha \alpha} \right\Vert_{2}^{2}. 
\end{align}

Similarly,
\begin{align}
 & \operatorname{Tr}\left( S_{AA'} \mathcal{E}^{\otimes 2} \left( S_{AA'} \right) \right) = \sum\limits_{i,j=1}^{d_{A}} \left\Vert \operatorname{Tr}_{B}\left( \mathcal{E} \left( | i \rangle \langle  j | \otimes I_{B} \right) \right) \right\Vert_{2}^{2} = \sum\limits_{i,j=1}^{d_{A}} \left\Vert \sum\limits_{\alpha}^{} \rho_{\alpha \alpha} \left\langle j | \rho_{\alpha \alpha} | i \right\rangle \right\Vert_{2}^{2} \\
&= \sum\limits_{i,j=1}^{d_{A}} \sum\limits_{\alpha, \beta}^{} \left\langle j | \rho_{\alpha \alpha} | i \right\rangle \left\langle i | \rho_{\beta \beta} | j\right\rangle \left\langle \rho_{\alpha \alpha}, \rho_{\beta \beta} \right\rangle = \sum\limits_{\alpha,\beta}^{} \left| \left\langle \rho_{\alpha \alpha}, \rho_{\beta \beta} \right\rangle \right|^{2}. 
\end{align}

Putting everything together, we have the desired result,
\begin{align}
  G(\mathcal{D}_{\mathbb{B}}) = \frac{1}{d^{2}} \left[ d_{B} \sum\limits_{\alpha}^{} \left\Vert \rho_{\alpha \alpha} \right\Vert_{2}^{2} - \sum\limits_{\alpha, \beta}^{} \left| \left\langle \rho_{\alpha \alpha}, \rho_{\beta \beta} \right\rangle \right|^{2}  \right] .
\end{align}

Define the renormalized Gram matrix as, \(X_{\alpha \beta} = \frac{\left\langle \rho_{\alpha \alpha}, \rho_{\beta \beta} \right\rangle}{d_{B}}\), then,
\begin{align}
  G(\mathcal{D}_{\mathbb{B}}) = \frac{1}{d_{A}^{2}} \left[ \sum\limits_{\alpha}^{} X_{\alpha \alpha} - \sum\limits_{\alpha \beta}^{} X_{\alpha \beta}^{2} \right] = \frac{1}{d_{A}^{2}} \left( \operatorname{Tr}\left( X \right) - \operatorname{Tr}\left( X^{2} \right) \right) = \frac{1}{d_{A}^{2}} \left\Vert X - X^{2} \right\Vert_{1} .
\end{align}

For the bound, note that \(X \geq X^{2}\) since \(X\) is bistochastic. Therefore, one has that \(\text{spec} \left( X \right) \subseteq \left[ 0,1 \right]  \). Then,
\begin{align}
  \left\Vert X - X^{2} \right\Vert_{1} = \sum\limits_{\alpha}^{} x_{\alpha} \left( 1 - x_{\alpha} \right) \leq \text{rank}(X)/4. 
\end{align}
And, \(\text{rank}(X) \leq \min \left( d_{A}^{2},d \right)\) since it is a Gram matrix of vectors in a \(d_{A}^{2}\)-dimensional space. Therefore, we have the bound,
\begin{align}
  G(\mathcal{D}_{\mathbb{B}}) \leq \frac{1}{4} \min \left( 1, \frac{d}{d_{A}^{2}} \right) = \frac{1}{4} \min \left( 1, \frac{d_{B}}{d_{A}} \right) .
\end{align}
\end{proof}

\subsection*{Proof of \autoref{th:deficit-entanglement}}
\deficitentanglement*
\begin{proof}
First notice, 
\begin{align}
X_{\alpha \beta} = \left\langle \frac{I_{A}}{d_{A}} + \Delta_{\alpha}, \frac{I_{B}}{d_{B}} + \Delta_{\beta} \right\rangle = \frac{1}{d_{A} d_{B}} + \frac{\left\langle \Delta_{\alpha}, \Delta_{\beta} \right\rangle}{d_{B}} \left( ~~\forall \alpha,\beta \right).
\end{align}

Namely, \(\hat{X}_{\mathbb{B}} = | \phi_{AB}^{s} \rangle \langle  \phi_{AB}^{s} | + \hat{\delta}_{\mathbb{B}} \) where \(| \phi_{AB}^{s} \rangle = \frac{1}{\sqrt{d_{A} d_{B}}} \sum\limits_{i=1}^{d_{A}} \sum\limits_{j=1}^{d_{B}} | i \rangle \otimes | j \rangle\) and \([ \hat{\delta} ]_{\alpha \beta} =  \frac{\left\langle \Delta_{\alpha}, \Delta_{\beta} \right\rangle}{d_{B}} \).

Then, using,
\begin{align}
\sum\limits_{\alpha}^{} \Delta_{\alpha} = \sum\limits_{\alpha}^{} \rho_{\alpha} - d_{B} I_{A} = \operatorname{Tr}_{B}\left[ \sum\limits_{\alpha}^{} | \psi_{\alpha} \rangle \langle  \psi_{\alpha} |  \right] = d_{B} I_{A} - d_{B} I_{A} = 0,
\end{align}
we find that
\begin{align}
\frac{1}{d_{A}^{2}} \operatorname{Tr}\left[ \hat{X}_{\mathbb{B}} - \hat{X}^{2}_{\mathbb{B}} \right] = \frac{1}{d_{A}^{2}} \operatorname{Tr}\left[ \hat{\delta}_{\mathbb{B}} - \delta_{\mathbb{B}}^{2} \right] = G(\mathcal{D}_{\mathbb{B}}).
\end{align}

Ignoring the squared term, it follows that
\begin{align}
G(\mathcal{D}_{\mathbb{B}}) \leq \frac{1}{d_{A}^{2}} \operatorname{Tr}\left[ \hat{\delta}_{\mathbb{B}} \right] = \frac{1}{d_{A}^{2}} \sum\limits_{\alpha}^{} \frac{\left\langle \Delta_{\alpha}, \Delta_{\alpha} \right\rangle}{d_{B}} \leq \frac{1}{d_{A}^{2}} \sum\limits_{\alpha}^{} \frac{\epsilon}{d_{B}} = \epsilon \frac{d}{d_{A}^{2} d_{B}} = \frac{\epsilon}{d_{A}}.
\end{align}
\end{proof}

\subsection*{Proof of \autoref{th:random-dephasing-channels}}
\randomdephasingchannels*
\begin{proof}
We have
\begin{align}
G(\mathcal{D}_{\mathbb{B}})=\frac{1}{d_A}\langle S_{AA'}, R_{\mathbb{B}}\rangle -\|R_{\mathbb{B}}^{AA'}\|_2^2,
\end{align}
Let us consider the two terms separately. 

\begin{align}
 \frac{1}{d_{A}} \operatorname{Tr}[ S_{AA'} U^{\otimes 2} \underbrace{ \left( \frac{1}{d}\sum\limits_{\alpha}^{} | \alpha \rangle \langle  \alpha | ^{\otimes 2} \right)}_{\Omega_{0}} U^{\dagger \otimes 2} ].
\end{align}
Then, notice that \(\mathbb{E}_{U} \left[ U^{\otimes 2} \Omega_{0} U^{\dagger \otimes 2}\right] = \frac{I + S}{d(d+1)}\), hence,
\begin{align}\label{eq:1st}
\frac{1}{  d_A  d \left( d+1 \right)} \operatorname{Tr}\left[ S_{AA'} \left( I+S \right) \right]= \frac{1 + \frac{d_{B}}{d_{A}}}{d+1} \leq \frac{2 }{d_{A}^{2}}.
\end{align}

Second term. Using convexity, we have,
\begin{align}
\overline{\left\Vert  \operatorname{Tr}_{BB'}R_{\mathbb{B}} \right\Vert_{2}^{2}}^{\mathbb{B}} \geq \left\Vert \overline{ \operatorname{Tr}_{BB'}R_{\mathbb{B}}}^{\mathbb{B}} \right\Vert_{2}^{2}.
\end{align}

Recall that,
\begin{align}
\overline{ \operatorname{Tr}_{BB'}R_{\mathbb{B}}}^{\mathbb{B}} = \operatorname{Tr}_{BB'}\left[ \frac{I+S}{d(d+1)} \right] = \frac{d_{B}^{2} I_{AA'} + d_{B} S_{AA'}}{d(d+1)}.
\end{align}

Therefore,
\begin{align}\label{eq:2nd}
\left\Vert \overline{R_{\mathbb{B}}}^{\mathbb{B}} \right\Vert_{2}^{2} = \frac{d_{B}^{2}}{d^2\left( d+1 \right)^{2}} \left\Vert d_{B} I_{AA'} + S_{AA'} \right\Vert_{2}^{2} = \frac{d_{B}^{2}}{d^2\left( d+1 \right)^{2}} \left[ d^{2} + d_{A}^{2} + 2d \right] \geq \frac{d_{B}^{2} }{\left( d+1 \right)^{2}}.
\end{align}

And finally, putting Eqs. (\ref{eq:1st}) and (\ref{eq:2nd}) together
\begin{align}
\overline{G(\mathcal{D}_{\mathbb{B}})}^{\mathbb{B}}  \leq \frac{2}{d_{A}^{2}} - \frac{d_{B}^{2}}{ \left( d+1 \right)^{2	}} \le \frac{2}{d_{A}^{2}} - \frac{d_{B}^{2}}{ \left( 2d \right)^{2	}}=\frac{7}{4}d_A^{-2}=O(\frac{1}{d_{A}^{2}}).
\end{align}

ii) 
\begin{align}
f(\mathbb{B}) \equiv \frac{1}{d_{A}} \left\langle S_{AA'}, R_{\mathbb{B}} \right\rangle - \left\Vert R_{\mathbb{B}}^{AA'} \right\Vert_{2}^{2} \equiv \alpha(\mathbb{B}) + \beta(\mathbb{B}).
\end{align}

We first collect a few results. First,
\begin{align}
\left| \alpha(\mathbb{B}) - \alpha (\widetilde{\mathbb{B}}) \right| = \frac{1}{d_{A}} \left| \left\langle S_{AA'}, R_{\mathbb{B}} - R_{\widetilde{\mathbb{B}}} \right\rangle \right| \leq \frac{1}{d_{A}} \left\Vert S_{AA'} \right\Vert_{\infty}^{} \left\Vert R_{\mathbb{B}} - R_{\widetilde{\mathbb{B}}} \right\Vert_{1}^{} \leq \frac{1}{d_{A}} \left\Vert R_{\mathbb{B}} - R_{\widetilde{\mathbb{B}}} \right\Vert_{1}^{} ,  
\end{align}
where in the first inequality we have used the Holder-type inequality (for matrices), \(\left| \operatorname{Tr}\left[ A^{\dagger}B \right] \right| \leq \left\Vert A \right\Vert_{\infty}^{} \left\Vert B \right\Vert_{1}^{}\). And in the second inequality, \(\left\Vert U \right\Vert_{\infty}^{} = 1\) for any unitary \(U\).

Second,
\begin{align}
&\left| \beta(\mathbb{B}) - \beta(\widetilde{\mathbb{B}}) \right| = \left| \left\Vert R_{\mathbb{B}}^{AA'} \right\Vert_{2}^{2} - \left\Vert R_{\widetilde{\mathbb{B}}}^{AA'} \right\Vert_{2}^{2} \right| = \left| \left( \left\Vert R_{\mathbb{B}}^{AA'} \right\Vert_{2} + \left\Vert R_{\widetilde{\mathbb{B}}}^{AA'} \right\Vert_{2} \right) \left( \left\Vert R_{\mathbb{B}}^{AA'} \right\Vert_{2} - \left\Vert R_{\widetilde{\mathbb{B}}}^{AA'} \right\Vert_{2} \right) \right|\\
&= 2 \left\Vert R_{\mathbb{B}}^{AA'} - R_{\widetilde{\mathbb{B}}}^{AA'} \right\Vert_{2} \leq 2 \left\Vert R_{\mathbb{B}}^{AA'} - R_{\widetilde{\mathbb{B}}}^{AA'} \right\Vert_{1} \leq 2 \left\Vert R_{\mathbb{B}} - R_{\widetilde{\mathbb{B}}} \right\Vert_{1},
\end{align}
where in the first inequality we have bounded the \(2\)-norm with the \(1\)-norm distance and in the second inequality we have used the fact that partial trace is a CP map and the \(1\)-norm is contractive under CP maps.

Now, we have to bound,
\begin{align}
\left\Vert R_{\mathbb{B}} - R_{\widetilde{\mathbb{B}}} \right\Vert_{1}^{} = \left\Vert R_{\mathbb{B}_{0}} - \left( V^{\dagger}U \right)^{\otimes 2} R_{\mathbb{B}_{0}} \left( V^{\dagger} U \right)^{\otimes 2} \right\Vert_{1}^{},  
\end{align}
where we have use the unitary invariance of the \(1\)-norm.

Define, \(U-V \equiv \Delta \implies V^{\dagger}U = I + \Delta\). Then,
\begin{align}
\left( I + \Delta \right)^{\otimes 2} = I \otimes I + \Delta \otimes I + I \otimes \Delta + \Delta \otimes \Delta \equiv I + X.
\end{align}

Using this, we have,
\begin{align}
\left\Vert R_{\mathbb{B}} - R_{\widetilde{B}} \right\Vert_{1}^{} = \left\Vert X R_{\mathbb{B}_{0}} + R_{\mathbb{B}_{0}} X + X R_{\mathbb{B}_{0}} X \right\Vert_{1}^{} \leq 2 \left\Vert X \right\Vert_{\infty}^{} + \left\Vert X \right\Vert_{\infty}^{2} = \left\Vert X \right\Vert_{\infty}^{} \left( 2 + \left\Vert X \right\Vert_{\infty}^{} \right),    
\end{align}
where we have repeatedly used \(\left\Vert AB \right\Vert_{1}^{} \leq \left\Vert A \right\Vert_{\infty}^{} \left\Vert B \right\Vert_{1}^{}\), submultiplicativity of norms and the fact that \(R_{\mathbb{B}_{0}}\) is a quantum state, \(\left\Vert R_{\mathbb{B}_{0}} \right\Vert_{1}^{} = 1\).

Now,
\begin{align}
&\left\Vert X \right\Vert_{\infty}^{} = \left\Vert \Delta \otimes I + I \otimes \Delta + \Delta \otimes \Delta \right\Vert_{\infty}^{} \leq 2 \left\Vert \Delta \right\Vert_{\infty}^{} + \left\Vert \Delta \right\Vert_{\infty}^{2} \\
& = \left\Vert \Delta \right\Vert_{\infty}^{} \left( 2 + \left\Vert \Delta \right\Vert_{\infty}^{}  \right) \leq 4 \left\Vert \Delta \right\Vert_{\infty}^{} = 4 \left\Vert U-V \right\Vert_{\infty}^{}.    
\end{align}

Therefore,
\begin{align}
\left\Vert R_{\mathbb{B}} - R_{\widetilde{\mathbb{B}}} \right\Vert_{1}^{} \leq 4 \left\Vert \Delta \right\Vert_{\infty}^{} \left( 2 + 4 \left\Vert \Delta \right\Vert_{\infty}^{}  \right) \leq 4 \left\Vert \Delta \right\Vert_{\infty}^{} \left( 2 + 4 \times 2 \right) = 40 \left\Vert \Delta \right\Vert_{\infty}^{},   
\end{align}
where we have used \(\left\Vert \Delta \right\Vert_{\infty}^{} \leq 2\).

Bringing everything together, we have,
\begin{align}
&\left| F(\mathbb{B}) - F(\widetilde{\mathbb{B}}) \right| \leq \left| \alpha(\mathbb{B}) - \alpha (\widetilde{\mathbb{B}}) \right| + \left| \beta(\mathbb{B}) - \beta (\widetilde{\mathbb{B}}) \right| \leq \left( \frac{1}{d_{A}} + 2 \right) \left\Vert R_{\mathbb{B}}- R_{\widetilde{\mathbb{B}}} \right\Vert_{1}^{}  \leq 40 \left( \frac{1}{d_{A}} + 2 \right) \left\Vert \Delta \right\Vert_{\infty}^{} \\
& \leq 40 \times \frac{5}{2} \left\Vert \Delta \right\Vert_{\infty}^{} = 100 \left\Vert U-V \right\Vert_{\infty}^{}  \leq 100 \left\Vert U-V \right\Vert_{2}^{}. 
\end{align}

\end{proof}

\subsection*{Proof of \autoref{th:entanglement-breaking-channel}}
\entanglementbreakingchannel*
\begin{proof}

To compute the open OTOC for the general case, \(G(\Phi_{\mathrm{EB}})\), we need to compute, \(\operatorname{Tr}\left( S \Phi^{\otimes 2}_{\mathrm{EB}} (S_{AA'}) \right) =\sum\limits_{i,j=1}^{d_{A}} \left\Vert \Phi_{\mathrm{EB}} \left( | i \rangle \langle  j | \otimes I_{B} \right) \right\Vert_{2}^{2}\) and \(\operatorname{Tr}\left( S_{AA'} \Phi^{\otimes 2}_{\mathrm{EB}} \left( S_{AA'} \right) \right) = \sum\limits_{i,j=1}^{d_{A}} \left\Vert \operatorname{Tr}_{B}\left( \Phi_{\mathrm{EB}} \left( | i \rangle \langle  j | \otimes I_{B} \right) \right) \right\Vert_{2}^{2}\).

\begin{align}
&(i)\quad \Phi_{\mathrm{EB}} \left( | i \rangle \langle  j | \otimes I_{B} \right) = \sum\limits_{k}^{} M_{k} \operatorname{Tr}\left[ \delta_{k} | i \rangle \langle  j | \otimes I_{B} \right] = \sum\limits_{k}^{} M_{k} \langle j | \delta_{k}^{A} |  i \rangle,
\end{align}
where \(\delta_{k}^{A} \equiv \operatorname{Tr}_{B}\left[ \delta_{k} \right]\).

Therefore,
\begin{align}
&\sum\limits_{i,j=1}^{d_{A}} \left\Vert \Phi_{\mathrm{EB}} \left( | i \rangle \langle  j | \otimes I_{B} \right) \right\Vert_{2}^{2} = \sum\limits_{i,j=1}^{d_{A}} \operatorname{Tr}\left[ \sum\limits_{k,k'}^{} M_{k} M_{k'} \langle i | \delta_{k}^{A} |  j \rangle \langle j | \delta_{k'}^{A} |  i \rangle \right]\\
&= \sum\limits_{k,k'}^{} \left\langle \delta_{k}^{A}, \delta_{k'}^{A} \right\rangle \left\langle M_{k},M_{k'} \right\rangle.
\end{align}

Similarly,
\begin{align}
&(ii)\quad \operatorname{Tr}_{B}\left[ \Phi_{\mathrm{EB}} \left( | i \rangle \langle  j | \otimes I_{B} \right) \right] = \operatorname{Tr}_{B}\left[ \sum\limits_{k}^{} M_{k} \operatorname{Tr}\left[ \delta_{k} | i \rangle \langle  j | \otimes I_{B} \right] \right] = \sum\limits_{k}^{} M^{A}_{k} \langle j | \delta_{k}^{A} |  i \rangle,
\end{align}
where \(M_{k}^{A} \equiv \operatorname{Tr}_{B}\left[ M_{k} \right]\).

Therefore,
\begin{align}
&\sum\limits_{i,j=1}^{d_{A}} \left\Vert \operatorname{Tr}_{B}\left[  \Phi_{\mathrm{EB}} \left( | i \rangle \langle  j | \otimes I_{B} \right)  \right] \right\Vert_{2}^{2} = \sum\limits_{i,j=1}^{d_{A}} \operatorname{Tr}\left[ \sum\limits_{k,k'}^{} M_{k}^{A} M_{k'}^{A} \langle i | \delta_{k}^{A} |  j \rangle \langle j | \delta_{k'}^{A} |  i \rangle \right]\\
&= \sum\limits_{k,k'}^{} \left\langle \delta_{k}^{A}, \delta_{k'}^{A} \right\rangle \left\langle M_{k}^{A},M_{k'}^{A} \right\rangle.
\end{align}

Putting everything together, we have,
\begin{align}
  G(\Phi_{\mathrm{EB}}) =  \frac{1}{d^{2}} \sum\limits_{k,k'}^{}  \left\langle \delta_{k}^{A}, \delta_{k'}^{A} \right\rangle \left[ d_{B} \left\langle M_{k}, M_{k'} \right\rangle -  \left\langle M_{k}^{A}, M_{k'}^{A} \right\rangle  \right],
\end{align}

\end{proof}
\subsection*{Proof of \autoref{th:gen-dephasing-channel}}
\generaldephasingchannel*
\begin{proof}
To prove (i), we need to show that given, \(\mathbb{B}=\{ | \alpha \rangle \}\) a basis of \(\mathcal{H}_{AB}\), \(d=\mathrm{dim}(\mathcal{H}_{AB})\) with \(\hat{\Phi} = \left[ \phi_{\alpha,\alpha'} \right]_{\alpha,\alpha'=1}^{d}\) such that \(\hat{\Phi} \geq 0\) and \(\phi_{\alpha,\alpha}=1 ~~\forall \alpha\), the map \(\mathcal{E}_{\hat{\Phi}}(X) = \sum\limits_{\alpha,\alpha'}^{d} \phi_{\alpha,\alpha'} X_{\alpha,\alpha'} | \alpha \rangle \langle  \alpha' | \) is a quantum channel.

First, notice that \(\mathcal{E}_{\hat{\Phi}}\) defines a linear map on \(\mathcal{L}(\mathcal{H}_{AB})\) such that
\begin{align}
\operatorname{Tr}\left[ \mathcal{E}_{\hat{\Phi}}(X)\right] = \sum\limits_{\alpha,\alpha'}^{} \phi_{\alpha,\alpha'} X_{\alpha,\alpha'} \delta_{\alpha,\alpha'} = \sum\limits_{\alpha}^{} \phi_{\alpha,\alpha} X_{\alpha,\alpha} = \sum\limits_{\alpha}^{} X_{\alpha} = \operatorname{Tr}\left[ X \right].
\end{align}
Hence, \(\mathcal{E}_{\hat{\Phi}}\) is a trace-preserving map.

Then, since \(\hat{\Phi} \geq 0\), one can write, \(\hat{\Phi} = S \hat{\Phi}_{D} S^{\dagger}\) where \(\hat{\Phi}_{D} \equiv \mathrm{diag}(\phi_{\mu}), \phi_{\mu} \geq0\) and \(S\) is a unitary. Then, \(\mathcal{E}_{\hat{\Phi}}\) can be expressed as,
\begin{align}
\mathcal{E}_{\hat{\Phi}}(X) = \sum\limits_{\mu,\alpha,\alpha'}^{} \lambda_{\mu} S_{\alpha, \mu} \overline{S_{\alpha', \mu}} X_{\alpha,\alpha'} | \alpha \rangle \langle  \alpha' | .
\end{align}
We now define \(A_{\mu} | \alpha \rangle:= \sqrt{\lambda_{\mu}} S_{\alpha,\mu} | \alpha \rangle ~~\forall \alpha,\alpha'\). Therefore,
\begin{align}
\mathcal{E}_{\hat{\Phi}}(X) = \sum\limits_{\mu,\alpha,\alpha'}^{} X_{\alpha,\alpha'} A_{\mu} | \alpha \rangle \langle  \alpha' | A_{\mu}^{\dagger} = \sum\limits_{\mu}^{} A_{\mu} X A^{\dagger}_{\mu}.
\end{align}
Moreover, 
\begin{align}
\langle \alpha | \sum\limits_{\mu}^{} A^{\dagger}_{\mu} A_{\mu} |\alpha' \rangle = \sum\limits_{\mu}^{} \lambda_{\mu} \overline{S_{\alpha,\mu}} S_{\alpha',\mu} \langle \alpha|\alpha' \rangle = \delta_{\alpha,\alpha'} \sum\limits_{\mu}^{} \lambda_{\mu} \left| S_{\alpha,\mu} \right|^{2} = \phi_{\alpha,\alpha} \delta_{\alpha,\alpha} = \delta_{\alpha,\alpha} \left( \forall \alpha,\alpha' \right).	
\end{align}
Therefore, \(\sum\limits_{\mu}^{}A^{\dagger}_{\mu} A_{\mu} = \mathbb{I}\) and since \(\mathcal{E}_{\hat{\Phi}}\) can be expressed in a Kraus form, it is CP.

\underline{Remark:} Let \(\mathcal{F} = \{ \hat{\Phi} \in \mathcal{M}_{d}^{\mathbb{C}} ~|~ \hat{\Phi} \geq 0 \text{ and } \phi_{\alpha,\alpha}=1 (\forall \alpha)  \}\). Then, \(\mathcal{F}\) is a convex subset of \(\mathcal{M}_{d}^{\mathbb{C}}\), the set of \(d \times d\) matrices over \(\mathbb{C}\). Maps of the form \(\mathcal{E}_{\hat{\Phi}}\) are parametrized by elements in \(\mathcal{F}\) and bases \(\mathbb{B}\). For a fixed \(\mathbb{B}\), the  map, \(\hat{\Phi} \in \mathcal{F} \rightarrow \mathcal{E}_{\hat{\Phi}}\) is an affine map of convex bodies.

\begin{center}
\noindent\rule[0.5ex]{0.5\linewidth}{0.8pt}
\end{center}

To prove (ii), the proof strategy is similar to Proposition 4. The key observation is that the action of the map, \(\mathcal{E}_{\hat{\Phi}}\) can be expressed as,
\begin{align}
\mathcal{E}_{\hat{\Phi}}(X) = \sum\limits_{\alpha,\beta=1}^{d} \phi_{\alpha,\beta} | \alpha \rangle \langle  \alpha | X | \beta \rangle \langle  \beta | = \sum\limits_{\alpha,\beta}^{d} \phi_{\alpha,\beta} x_{\alpha,\beta} | \alpha \rangle \langle  \beta | ,
\end{align}
where \(x_{\alpha,\beta} \equiv \langle \alpha | X |  \beta \rangle\). This follows from the action \(\mathcal{E}_{\hat{\Phi}}(| \alpha \rangle \langle  \alpha' | ) = \phi_{\alpha,\alpha'} | \alpha \rangle \langle  \alpha' | \).

We need to evaluate \(\operatorname{Tr}\left( S \mathcal{E}^{\otimes 2} (S_{AA'}) \right) =\sum\limits_{i,j=1}^{d_{A}} \left\Vert \mathcal{E} \left( | i \rangle \langle  j | \otimes I_{B} \right) \right\Vert_{2}^{2}\) and \(\operatorname{Tr}\left( S_{AA'} \mathcal{E}^{\otimes 2} \left( S_{AA'} \right) \right) = \sum\limits_{i,j=1}^{d_{A}} \left\Vert \operatorname{Tr}_{B}\left( \mathcal{E} \left( | i \rangle \langle  j | \otimes I_{B} \right) \right) \right\Vert_{2}^{2}\).

Now, 
\begin{align}
\mathcal{E} \left( | i \rangle \langle  j | \otimes I_{B} \right) = \sum\limits_{\alpha,\beta}^{} \phi_{\alpha,\beta} | \alpha \rangle \langle  \alpha | \left( | i \rangle \langle  j | \otimes I_{B} \right) | \beta \rangle \langle  \beta | = \sum\limits_{\alpha,\beta}^{} \phi_{\alpha,\beta} \Pi_{\alpha,\beta} \operatorname{Tr}\left[ \rho_{\alpha,\beta} | i \rangle \langle  j |  \right], 
\end{align}
where \(\Pi_{\alpha,\beta} \equiv | \alpha \rangle \langle  \beta |\) and \(\rho_{\alpha,\beta} \equiv \operatorname{Tr}_{B}\left[ | \alpha \rangle \langle  \beta |  \right]\).

And,
\begin{align}
&\sum\limits_{i,j=1}^{d_{A}} \left\Vert \mathcal{E} \left( | i \rangle \langle  j | \otimes I_{B} \right) \right\Vert_{2}^{2} = \sum\limits_{i,j}^{d_{A}} \operatorname{Tr}\left[ \sum\limits_{\alpha,\beta,\gamma,\delta}^{} \phi_{\alpha,\beta}^{*} \phi_{\gamma,\delta} \overline{\langle j | \rho_{\alpha,\beta} |  i \rangle} \Pi_{\beta,\alpha} \langle j | \rho_{\gamma,\delta} |  i \rangle \Pi_{\gamma,\delta} \right]\\
&=\sum\limits_{i,j}^{d_{A}} \sum\limits_{\alpha,\beta,\gamma,\delta}^{} \phi_{\alpha,\beta}^{*} \phi_{\gamma,\delta} \overline{\langle j | \rho_{\alpha,\beta} |  i \rangle} \underbrace{\operatorname{Tr}\left[\Pi_{\beta,\alpha} \langle j | \rho_{\gamma,\delta} |  i \rangle \Pi_{\gamma,\delta} \right]}_{=\delta_{\alpha,\gamma} \delta_{\beta,\delta}}\\
& = \sum\limits_{i,j}^{d_{A}} \sum\limits_{\alpha,\beta}^{} \left| \phi_{\alpha,\beta} \right|^{2} \left| \langle j | \rho_{\alpha,\beta} |  i \rangle \right|^{2} = \sum\limits_{\alpha,\beta}^{} \left| \phi_{\alpha,\beta} \right|^{2} \left\Vert \rho_{\alpha,\beta} \right\Vert_{2}^{2}.
\end{align}

Similarly, we have,
\begin{align}
\operatorname{Tr}_{B}\left[ \mathcal{E} \left( | i \rangle \langle  j | \otimes I_{B} \right) \right] = \sum\limits_{\alpha,\beta}^{} \phi_{\alpha,\beta} \langle j | \rho_{\alpha,\beta} |  i \rangle \underbrace{\operatorname{Tr}_{B}\left[ \Pi_{\alpha,\beta} \right]}_{=\rho_{\alpha,\beta}} = \sum\limits_{\alpha,\beta}^{} \phi_{\alpha,\beta} \langle j | \rho_{\alpha,\beta} |  i \rangle \rho_{\alpha,\beta}.
\end{align}

Then,
\begin{align}
\left\Vert \operatorname{Tr}_{B}\left[ \mathcal{E} \left( | i \rangle \langle  j | \otimes I_{B} \right) \right] \right\Vert_{2}^{2} = \sum\limits_{\alpha,\beta,\gamma,\delta}^{} \phi_{\alpha,\beta}^{*} \phi_{\gamma,\delta} \langle i | \rho_{\beta,\alpha} |  j \rangle \langle j | \rho_{\gamma,\delta} |  i \rangle \operatorname{Tr}\left[ \rho_{\beta,\alpha} \rho_{\gamma,\delta} \right].
\end{align}
And,
\begin{align}
&\sum\limits_{i,j=1}^{d_{A}}\left\Vert \operatorname{Tr}_{B}\left[ \mathcal{E} \left( | i \rangle \langle  j | \otimes I_{B} \right) \right] \right\Vert_{2}^{2} =  \sum\limits_{i,j=1}^{d_{A}} \sum\limits_{\alpha,\beta,\gamma,\delta}^{} \phi_{\alpha,\beta}^{*} \phi_{\gamma,\delta} \langle i | \rho_{\beta,\alpha} |  j \rangle \langle j | \rho_{\gamma,\delta} |  i \rangle \operatorname{Tr}\left[ \rho_{\beta,\alpha} \rho_{\gamma,\delta} \right]\\
& = \sum\limits_{\alpha,\beta,\gamma,\delta}^{} \phi_{\alpha,\beta}^{*} \phi_{\gamma,\delta} \left| \left\langle \rho_{\alpha,\beta}, \rho_{\gamma,\delta} \right\rangle \right|^{2}.
\end{align}

Putting everything together, we have the desired proof.
\end{proof}

\end{document}